\documentclass[11pt]{article}
\usepackage[margin=1in]{geometry}
\usepackage{amsmath,amssymb}
\usepackage{graphicx}
\graphicspath{{figs/}}
\usepackage{mathtools}
\usepackage{amsthm}
\usepackage[T1]{fontenc}
\newtheorem{proposition}{Proposition}
\theoremstyle{definition}
\newtheorem*{hypothesis}{Standing Hypothesis (H)}

\usepackage[hidelinks]{hyperref}

\newcommand{\AAA}{\text{\AA}}
\makeatletter
\let\saved@AA\AA
\renewcommand{\AA}{\text{\saved@AA}}
\makeatother
\title{A Three-Degree-of-Freedom Chesnavich Model for Roaming:\\
Derivation, Phase-Space Geometry, NHIM-Anchored Dividing Surfaces, and Roaming Transport}
\author{S. Wiggins\\[2pt]
\small Hetao Institute of Mathematics and Interdisciplinary Sciences, Shenzhen, China\\
\small School of Mathematics, University of Bristol, Bristol BS8 1TW, United Kingdom}
\date{\today}

\begin{document}
\maketitle

\begin{abstract}
Roaming reactions, in which a dissociating fragment moves through a flat region of the potential surface rather than down the minimum-energy path, lie outside the assumptions of conventional transition state theory. The phase-space theory of roaming --- unstable periodic orbits and their invariant manifolds organizing transport --- has been developed for the Chesnavich model of $\mathrm{CH_4^+}\to\mathrm{CH_3^+}+\mathrm{H}$, which is cylindrically symmetric and reduces to two degrees of freedom (2-DoF). We construct and analyze a three-degree-of-freedom (3-DoF) extension. From the rigid-body formulation of Ezra and Wiggins, we break the symmetry with an azimuthal coupling respecting the three-fold ($C_3$) symmetry of the methyl fragment, obtaining a family $H_b$ whose planar reduction at $b=0$ is the 2-DoF model exactly and which is genuinely 3-DoF for $b>0$. This activates the out-of-plane degree of freedom at once: with the physical planar-top inertia ratio $I_z=2I_x$, arbitrarily weak coupling makes the periodic orbit on the roaming shelf transversely unstable, opening an escape route out of the reaction plane. Apart from a narrow elliptic window $0.58\lesssim b\lesssim0.63$, the instability persists across the range studied, changing type through a period-doubling at $b_c\approx0.63$. Because a periodic orbit cannot anchor a dividing surface in three degrees of freedom, we construct the objects that do --- three three-dimensional normally hyperbolic invariant manifolds, one per transition state --- at $b=0$, and prove that every compact interior piece of each persists for sufficiently small $b>0$. At $E=0.5\ \mathrm{kcal\,mol^{-1}}$ the coupling lowers the direct non-reactive fraction of a microcanonical ensemble of incoming trajectories by $0.032$ and raises the two roaming fractions by $0.040$; the effect decreases as the energy increases.
\end{abstract}

\section{Introduction}
In conventional transition state theory (TST) a reaction is controlled by a single recrossing-free dividing surface separating reactants from products \cite{TruhlarGarrettKlippenstein1996,WaalkensSchubertWiggins2008}. Roaming reactions violate this picture: a dissociating fragment, instead of separating directly or following the minimum-energy path, moves through a flat region of the potential energy surface at wide amplitude before re-encountering its partner and reacting, often through an unexpected channel \cite{Townsend2004,BowmanSuits2011,Bowman2014,Mauguiere2017review}. Roaming has been identified in a broad range of systems \cite{Suits2020}.

The phase-space (geometric) theory of reaction dynamics provides a framework for roaming. The objects that organize transport are not critical points of the potential but invariant sets of the dynamics: unstable periodic orbits (POs) in two degrees of freedom, and normally hyperbolic invariant manifolds (NHIMs) in higher dimension, together with their stable and unstable manifolds \cite{WaalkensSchubertWiggins2008,Wiggins2016}. For the two-degree-of-freedom (2-DoF) Chesnavich model of the ion--molecule reaction $\mathrm{CH_4^+}\to\mathrm{CH_3^+}+\mathrm{H}$ \cite{Chesnavich1986}, this program has been carried through in detail: the relevant POs, their bifurcations, and the resulting roaming pathways have been mapped as functions of energy and of the potential parameters \cite{Mauguiere2014,Mauguiere2016,KrajnakWiggins2018,EzraWiggins2019}.

Chesnavich's potential is cylindrically symmetric, so the azimuthal angular momentum is conserved and the model reduces to two degrees of freedom. Roaming in real systems has more degrees of freedom, and the phase-space theory of NHIMs has been extended to three degrees of freedom \cite{KrajnakGarciaGarridoWiggins2021}. No concrete, tractable 3-DoF Chesnavich-type model has been available on which to test and develop that theory. This paper supplies one.

Section~\ref{sec:model} derives the 3-DoF Chesnavich Hamiltonian from the Ezra--Wiggins rigid-body formulation \cite{EzraWiggins2019}, chooses the lowest-order smooth coupling compatible with the $C_3$ symmetry of the methyl fragment, and describes the model's physical and chemical content. Section~\ref{sec:diagnostics} develops the configuration-space picture of roaming, states how roaming is defined --- as a property of trajectories classified against three transition states --- and gives the dimensional reason a periodic orbit can no longer anchor a dividing surface in three degrees of freedom. Section~\ref{sec:results} presents the phase-space geometry as a function of the coupling $b$ and the energy $E$. The central object is the periodic orbit FR1, the unstable orbit lying on the roaming shelf --- the flat region of the potential where roaming occurs --- whose dividing surface classifies trajectories as roaming or direct in the 2-DoF theory; the central finding is that FR1 becomes unstable in the out-of-plane direction as soon as the cylindrical symmetry is broken. Section~\ref{sec:dimensionality} compares this with the formaldehyde--acetaldehyde contrast and shows, through the growth of out-of-plane motion along trajectories, when a two-degree-of-freedom description suffices and when it fails. Section~\ref{sec:nhim} constructs the three three-dimensional NHIMs explicitly in the $b=0$ limit and establishes three things about them: the topology of the dividing surface each anchors, the status of FR1 as the limiting planar member of its manifold, and the persistence of their compact interior pieces for sufficiently small $b>0$. These results are stated and proved as six propositions (Section~\ref{sec:proofs}) resting on a single numerically verified hypothesis, the existence and smoothness of the reduced periodic-orbit family. Section~\ref{sec:transport} carries the trajectory classification into three dimensions and computes the class fractions and gap-time statistics as functions of $b$ and energy. Section~\ref{sec:discussion} discusses the chemical and dynamical significance, and Section~\ref{sec:conclusion} concludes. Appendices document the numerical methods and the technical issues encountered, including a coordinate singularity intrinsic to body-frame spherical coordinates.

\begin{figure}[t]
\centering
\includegraphics[width=\textwidth]{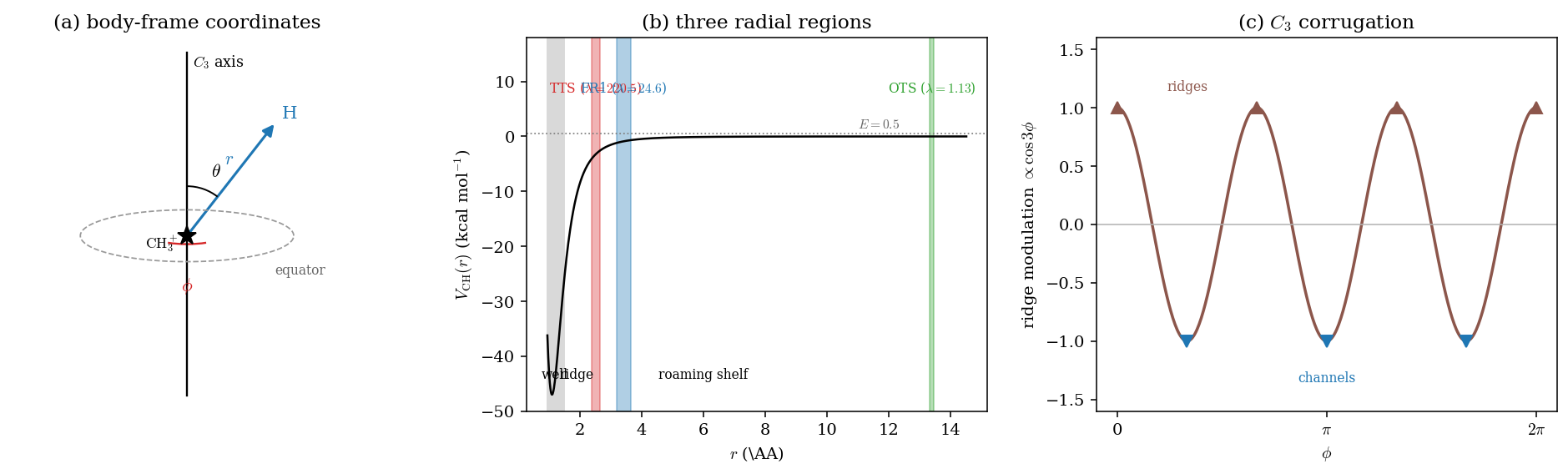}
\caption{Orientation. (a) Body-frame coordinates: the departing H atom is located relative to the $\mathrm{CH_3^+}$ center of mass by the radial separation $r$, the polar (bending) angle $\theta$ measured from the $C_3$ symmetry axis (a co-latitude), and the azimuthal angle $\phi$ about that axis (a longitude). A \emph{meridional} plane is a half-plane containing the axis (a plane of constant $\phi$, like a line of longitude); the \emph{equator} is the plane $\theta=\pi/2$. (b) The three radial regions and, on the energy line $E=0.5\ \mathrm{kcal\,mol^{-1}}$ (dotted), the radial locations of the three planar orbits that generate the entrance, classifier, and reactive-gate NHIMs: the tight transition state TTS-PO (reactive gate, $\lambda=220.5$), arcing over the mouth of the well at $r\in[2.38,2.64]\ \AAA$, just outside the ridge crest; the roaming-shelf orbit FR1 (classifier, $\lambda=24.6$); and the orbiting transition state OTS-PO (entrance, $\lambda=1.13$) at the centrifugal barrier. Here $\lambda$ is the Floquet multiplier of the orbit --- the factor by which a small displacement from it grows over one period, defined in Section~\ref{sec:diagnostics} --- so the reactive gate is strongly unstable and the entrance only weakly so. Each is the limiting planar (separatrix) member of a two-component NHIM family parametrized by the angular momentum (Section~\ref{sec:nhim}). (c) The $C_3$ azimuthal \emph{corrugation} introduced by the coupling: the bending ridge is rippled three-fold in azimuth, $\propto b\cos3\phi$, with three equivalent ridges and three intervening channels --- the imprint of the three hydrogens of the methyl group.}
\label{fig:coords}
\end{figure}

\section{The three-degree-of-freedom model}\label{sec:model}
\subsection{Coordinates and reduced Hamiltonian}
Following Ezra and Wiggins \cite{EzraWiggins2019}, the Chesnavich system is modeled as a rigid symmetric top --- the $\mathrm{CH_3^+}$ fragment, with principal moments of inertia $I_x=I_y\neq I_z$ --- coupled to a structureless particle representing the departing H atom. The H atom's position relative to the fragment center of mass is given in molecule-fixed spherical coordinates $(r,\theta,\phi)$, where $r$ is the radial separation, $\theta$ the polar (bending) angle measured from the $C_3$ symmetry axis, and $\phi$ the azimuthal angle about that axis (Fig.~\ref{fig:coords}a). The fragment orientation is described by Euler angles. The full molecular kinetic energy is the sum of the rigid-rotor energy of the top and the kinetic energy of the H atom.

We work at zero total angular momentum, $J=0$, the natural setting in which to isolate the internal roaming dynamics. In this case the Ezra--Wiggins reduction expresses the Hamiltonian in terms of $(r,\theta,\phi)$ and their conjugate momenta $(p_r,p_\theta,p_\phi)$ as
\begin{equation}\label{eq:H}
H = \frac{1}{2I_x}\left(p_\theta^2 + p_\phi^2\cot^2\theta\right) + \frac{p_\phi^2}{2I_z} + \frac{1}{2m}\left(p_r^2 + \frac{p_\theta^2}{r^2} + \frac{p_\phi^2}{r^2\sin^2\theta}\right) + V(r,\theta,\phi),
\end{equation}
where $m$ is the $\mathrm{H}$--$\mathrm{CH_3^+}$ reduced mass and the rotational kinetic energy uses $L_x^2+L_y^2=p_\theta^2+p_\phi^2\cot^2\theta$, $L_z=p_\phi$. This $\cot^2\theta$ form is the one given by the Ezra--Wiggins rigid-body reduction \cite{EzraWiggins2019}, which we follow throughout.

\subsection{Potential and the symmetry-breaking coupling}
Chesnavich's potential \cite{Chesnavich1986} consists of a radial term describing the C--H stretch and a hindered-rotor coupling describing the bend:
\begin{align}
V_{\mathrm{CH}}(r) &= \frac{D_e}{c_1-6}\left[2(3-c_2)\,e^{c_1(1-x)} - (4c_2-c_1 c_2 + c_1)\,x^{-6} - (c_1-6)c_2\,x^{-4}\right], \quad x=\frac{r}{r_e}, \label{eq:VCH}\\
V_0(r) &= V_e\,e^{-\alpha(r-r_e)^2}. \label{eq:V0}
\end{align}
In the original 2-DoF model the angular potential is $\tfrac12 V_0(r)(1-\cos2\theta)$, which depends only on the bending angle $\theta$ and not on the azimuth $\phi$; the model is therefore cylindrically symmetric, $p_\phi$ is conserved, and the dynamics reduces to two degrees of freedom. The standard parameters, fitted to $\mathrm{CH_4^+}$, are $D_e=47$, $r_e=1.1$, $c_1=7.37$, $c_2=1.61$, $V_e=55$ (energies in $\mathrm{kcal\,mol^{-1}}$, lengths in \AAA), $\alpha=1\,\AAA^{-2}$, $I_x=2.373409\ \mathrm{u\,\AAA^2}$, and the reduced mass $m=0.9445\ \mathrm{u}$, from $m_{\mathrm{H}}=1.007825\ \mathrm{u}$ and $m_{\mathrm{C}}=12.0\ \mathrm{u}$ \cite{Chesnavich1986,Mauguiere2014,EzraWiggins2019}. In these units (u, \AAA, $\mathrm{kcal\,mol^{-1}}$) the induced unit of time is $\AAA\sqrt{\mathrm{u}/(\mathrm{kcal\,mol^{-1}})}\approx48.9\ \mathrm{fs}$; all times below are quoted in this unit.

To obtain a genuine three-degree-of-freedom system we break the cylindrical symmetry by introducing azimuthal dependence. The admissible azimuthal couplings are constrained by the symmetry of the methyl fragment. $\mathrm{CH_3^+}$ has $D_{3h}$ symmetry; the three hydrogen atoms lie at $120^\circ$ intervals, so any potential the departing H atom experiences as a function of azimuth must be invariant under the three-fold rotation $\phi\to\phi+2\pi/3$. The lowest-order azimuthal harmonic with this property is $\cos3\phi$. A coupling built from $\cos2\phi$, by contrast, would impose a spurious two-fold symmetry inconsistent with the molecular point group. We therefore take
\begin{equation}\label{eq:Vcoup}
V_{\mathrm{coup}}(r,\theta,\phi) = V_0(r)\left[\sin^2\theta + b\,\sin^3\theta\,\cos3\phi\right],
\end{equation}
so that the total potential is $V=V_{\mathrm{CH}}+V_{\mathrm{coup}}$. The first term equals $\tfrac12 V_0(1-\cos2\theta)$ and reproduces the Chesnavich bend; we write it as $\sin^2\theta$ because its Cartesian form, $V_0\,(X^2+Y^2)/r^2$, is a ratio of polynomials whose denominator does not vanish away from the origin, so that smoothness across the symmetry axis is manifest for both angular terms. The second term is the symmetry-breaking term, with $b\in[0,1]$ a dimensionless coupling strength. The factor $\sin^3\theta\cos3\phi$ is the angular factor of the real solid harmonic $X^3-3XY^2$ of degree three and order three: in Cartesian molecule-fixed coordinates $(X,Y,Z)=(r\sin\theta\cos\phi,\,r\sin\theta\sin\phi,\,r\cos\theta)$ it is $(X^3-3XY^2)/r^3$, which is smooth everywhere except at the origin. This smoothness is essential: a one-fold coupling of the form $V_0(r)\sin^2\theta\cos\phi$ --- Cartesian form $V_0\,X\sqrt{X^2+Y^2}/r^2$ --- is only once continuously differentiable on the symmetry axis, where $\phi$ is undefined; its second derivatives, which enter the variational equations, are discontinuous there, and since the orbits constructed below cross that axis their linear stability would be ill-defined (Appendix~\ref{app:singularity}).

\subsection{Conserved quantities and symmetries}
The Hamiltonian and coupling introduced above---their reduction, at $b=0$ on the invariant set $p_\phi=0$, to the 2-DoF Chesnavich model, the smoothness of the coupling across the symmetry axis, and the discrete symmetries recorded next---were confirmed in a computer-algebra system (Appendix~\ref{app:symbolic}). Three properties of $H_b$ organize the analysis that follows: the exact reduction, on the invariant set $p_\phi=0$, to the 2-DoF Chesnavich model at $b=0$, which anchors every subsequent result to a known limit; the loss of the azimuthal integral for $b>0$, which makes the dynamics genuinely three-degree-of-freedom; and a set of discrete symmetries, which single out the invariant reaction plane in which the generating orbits lie and the transverse subspace whose stability governs out-of-plane escape.

\noindent\textbf{(i) Reduction at $b=0$.} For $b=0$ the Hamiltonian \eqref{eq:H} satisfies $\partial H/\partial\phi=0$, so $\phi$ is cyclic and $p_\phi$ is conserved. Restricting to the invariant set $p_\phi=0$, \eqref{eq:H} reduces identically to the 2-DoF Chesnavich Hamiltonian, with $I\to I_x$ and $\mu\to m$. The family $H_b$ is therefore an exact homotopy whose $b=0$ member carries Chesnavich's 2-DoF model as its planar reduction.

\noindent\textbf{(ii) Symmetry breaking at $b>0$.} $\partial H/\partial\phi=-3bV_0(r)\sin^3\theta\sin3\phi\neq0$, so $p_\phi$ is no longer conserved and the dynamics is genuinely three-degree-of-freedom. Energy is the only known global integral.

\noindent\textbf{(iii) Discrete symmetries.} The Hamiltonian is invariant under the canonical reflections $(\theta,p_\theta)\to(\pi-\theta,-p_\theta)$ and $(\phi,p_\phi)\to(-\phi,-p_\phi)$ (since $\cos3\phi$ is even), and under the $C_3$ rotation $\phi\to\phi+2\pi/3$. The reflection $(\phi,p_\phi)\to(-\phi,-p_\phi)$ fixes the set $\{Y=0,\ p_Y=0\}$ --- the reaction plane --- which is therefore invariant under the flow; the planar member of the FR1 family lies in the reaction plane, and out-of-plane perturbations $(Y,p_Y)$ form a well-defined transverse subspace whose stability we analyze below.

\subsection{Physical inertia of a planar top}\label{sec:inertia}
The moment-of-inertia ratio $I_z/I_x$ is not an incidental parameter of the model: it fixes the relative timescales of rotation about the symmetry axis and tumbling perpendicular to it, and it enters the transverse dynamics directly through the $p_\phi^2/(2I_z)$ term in \eqref{eq:H}. In a model constructed only for mathematical convenience one might be tempted to leave this ratio free, or to set $I_z=I_x$ so that the top is spherical and the rotational kinetic energy is isotropic. Either choice would be physically wrong for $\mathrm{CH_3^+}$, and, as Section~\ref{sec:results} shows, either would change the transverse stability of the orbit qualitatively; it is therefore essential that the value used be the one dictated by the geometry of the fragment rather than a tunable knob.

For a planar (laminar) rigid body the perpendicular-axis theorem states that the moment of inertia about the axis normal to the plane equals the sum of the moments about any two orthogonal in-plane axes through the same point \cite{GoldsteinPooleSafko2002}. $\mathrm{CH_3^+}$ is planar, and its three-fold symmetry makes the two in-plane principal moments equal, $I_x=I_y$, so the theorem gives $I_z=I_x+I_y=2I_x$ exactly: the ion is a planar oblate top with $I_z/I_x=2$. With $I_x=2.373409\ \mathrm{u\,\AAA^2}$ \cite{Chesnavich1986,EzraWiggins2019} this gives $I_z=4.746818\ \mathrm{u\,\AAA^2}$. We use this physical value throughout. As shown in Section~\ref{sec:results}, it is, together with the $C_3$ coupling, what produces the central bifurcation of the paper: the transverse direction is hyperbolic at weak coupling precisely because of the oblate ratio $I_z=2I_x$, and the unphysical spherical choice $I_z=I_x$ removes the effect.

\subsection{Chemical significance}\label{sec:chemsig}
The model represents the post-transition-state dynamics of $\mathrm{CH_4^+}\to\mathrm{CH_3^+}+\mathrm{H}$ once the departing H atom has reached the flat, long-range region of the potential where roaming occurs. Fragmentation of $\mathrm{CH_4^+}$ to $\mathrm{CH_3^+}+\mathrm{H}$ is an established decay channel of the methane cation --- seen, for instance, when a $\mathrm{CH_4^+}$ intermediate formed by charge transfer ejects a hydrogen atom \cite{MeyerWester2017} --- and the wide-amplitude, non-minimum-energy-path character of such dissociations is the hallmark of roaming \cite{Bowman2014,Mauguiere2017review}. The long-range ion--neutral interaction that governs this regime, comprising the ion-induced-dipole attraction and the higher anisotropic terms of the charge--molecule potential, is precisely the interaction whose central-field and orbiting-transition-state treatment underlies the phase-space theory of ion--molecule capture and reaction \cite{ChesnavichBowers1982,MeyerWester2017}.

The radial coordinate $r$ is the dissociating C--H distance; the bending angle $\theta$ measures the departure of the H atom from the $\mathrm{CH_3^+}$ symmetry axis; and the azimuthal angle $\phi$ measures rotation of the H atom about that axis relative to the methyl frame. $\mathrm{CH_3^+}$ is planar (point group $D_{3h}$), so the charge distribution it presents to the departing H atom carries an intrinsic three-fold azimuthal structure. In the cylindrically symmetric limit the azimuthal motion is decoupled by an exact integral: $p_\phi$ is conserved and each trajectory is confined to its $p_\phi$-leaf. The symmetry-breaking parameter $b$ measures the strength with which this three-fold structure couples to the azimuthal motion of the roaming H atom; physically it interpolates between an idealized axially averaged interaction ($b=0$) and one that resolves the discrete three-fold structure of the methyl group ($b>0$). That molecular orientation and geometry --- and not merely the stationary points of the potential --- can control the outcome of ion--molecule reactions is by now well documented \cite{MeyerWester2017}, which is what makes resolving this structure physically consequential. The question the model is designed to answer is how this three-fold azimuthal structure, present in the real molecule but suppressed in the 2-DoF model, reshapes the phase-space structures that control roaming.

\section{The geometry of roaming in three degrees of freedom}\label{sec:diagnostics}

Roaming is a reaction mechanism: a fragment leaves the well and, rather
than dissociating directly, moves through a flat part of the potential energy surface
before it reacts or separates. Its description in phase space rests on three parts of
the surface with distinct chemical meaning --- the deep well, the roaming shelf, and
the dissociation asymptote --- and on the transition states whose dividing surfaces
gate transport between them. This section fixes the three parts of the surface and
their chemical content; the transition states of Section~\ref{sec:po} are then
interpreted as the gates between them. Two features distinguish the 3-DoF surface
from the 2-DoF model: $V(r,\theta,\phi)$ is defined on a three-dimensional
configuration space, so no single planar contour plot represents it, and the
structures that gate transport are no longer periodic orbits
(Section~\ref{sec:po}).

\subsection{The configuration-space landscape and the meaning of roaming}\label{sec:landscape}

In the 2-DoF Chesnavich model the three parts of the surface are identified from the
contours of $V(r,\theta)$ \cite{Wiggins2026,Mauguiere2014}. At small $r$ a deep
well holds the H atom bound to the fragment. At intermediate-to-large $r$ the
Chesnavich lock $V_0(r)$ has decayed and the angular potential is shallow, so the H
atom moves through wide ranges of $\theta$ at nearly constant radial energy. In the
radial coordinate this part of the surface is a ledge: $V_{\mathrm{CH}}(r)$ rises
steeply out of the well and then levels off, sloping gently outward onto the
dissociation asymptote $V\to0$, so that $r$ varies by several~\AAA\ at nearly constant
potential energy (Fig.~\ref{fig:landscape}b, on-axis cut). We call this feature of
the potential the \emph{roaming shelf}, the term making explicit the radial flatness
that is the geometric condition for roaming: a trajectory can persist at large $r$
without descending a radial slope. Roaming is motion that, having left the well,
remains on the shelf rather than crossing inward to the well or outward to
dissociation; its operational definition --- a classification of trajectories by
their crossings of three dividing surfaces --- is given in Section~\ref{sec:po}.

The 3-DoF model has the same three parts. Because $V(r,\theta,\phi)$ is defined on a
three-dimensional configuration space, we represent it by two-dimensional cuts
(Fig.~\ref{fig:landscape}). At $b=0$ the potential is cylindrically symmetric, and a
single meridional cut --- the half-plane of constant $\phi$ containing the symmetry
axis, coordinates $(\rho,Z)$ with $\rho=\sqrt{X^2+Y^2}$ --- carries the whole
surface; this is the 2-DoF landscape of \cite{Mauguiere2014,Wiggins2026} (Fig.~\ref{fig:landscape}a). The well
sits on the axis at $r=r_e$, the collinear $\mathrm{CH_3^+}\cdots$H geometry at which
the bending term $V_0(r)\sin^2\theta$ vanishes. Away from the axis this term raises
the potential into a ridge at the equator $\theta=\pi/2$ --- the bending ridge ---
which hinders changes in the polar angle where $V_0(r)$ is appreciable. The radial
cuts of Fig.~\ref{fig:landscape}b make the energetics explicit: on the axis
($\theta=0$) the potential is $V_{\mathrm{CH}}(r)$ alone, the well with no bend; at
the equator ($\theta=\pi/2$) it carries the bending ridge; both flatten onto the
roaming shelf beyond $r\approx3$~\AAA, and FR1 lies at $r\approx3.2$--$3.65$~\AAA\
(Section~\ref{sec:po}).

The 3-DoF content absent from the 2-DoF model is the dependence on the azimuth
$\phi$. For $b>0$ the coupling $b\,V_0(r)\sin^3\theta\cos3\phi$ modulates the bending
ridge three-fold in $\phi$: at the equator it adds $b\,V_0(r)\cos3\phi$ to the ridge
height, lowering the ridge in three channels near $\phi=\pi/3,\ \pi,\ 5\pi/3$ (where
$\cos3\phi=-1$) and raising it in the three directions between. The bending ridge ---
the barrier between the well and the roaming shelf --- is therefore corrugated
three-fold around the axis, following the three hydrogens of the methyl fragment;
$b$ is the depth of the corrugation. Because the modulation is carried by $V_0(r)$,
it is largest at the inner edge of the shelf, where $V_0$ is of order the roaming
energy, and vanishes as $V_0\to0$ at large $r$. Figure~\ref{fig:landscape}c shows
the modulation on the sphere of directions $(\theta,\phi)$ at $r=3.1$~\AAA\ for
$b=0.8$.

\begin{figure}[t]
\centering
\includegraphics[width=\textwidth]{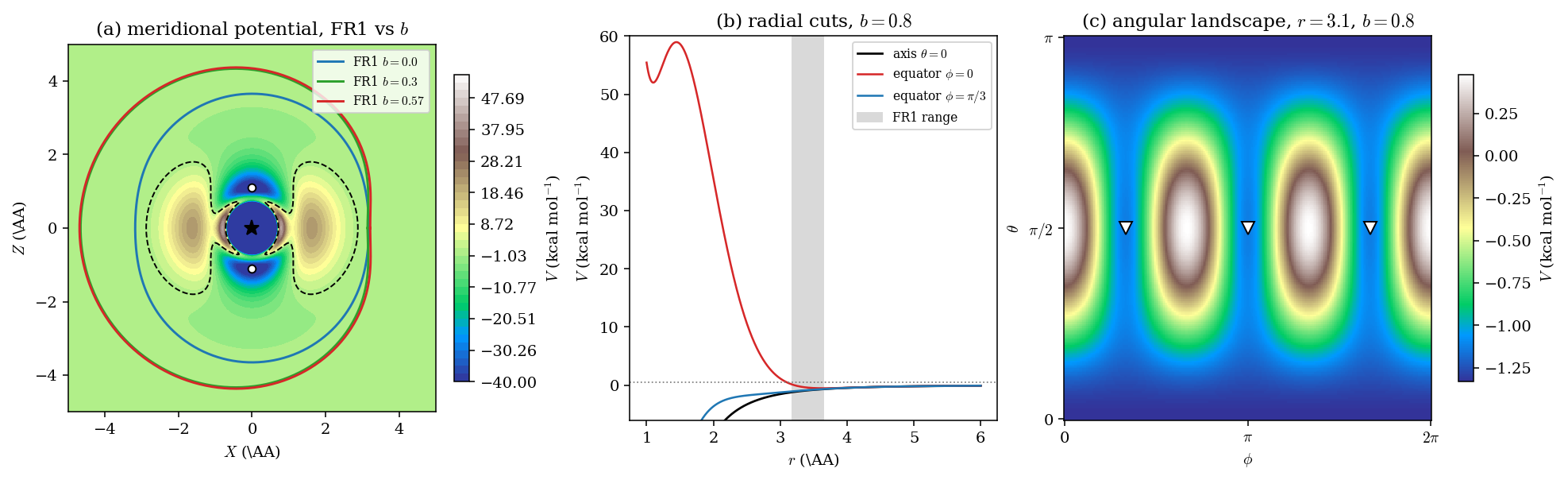}
\caption{The configuration-space landscape of the 3-DoF Chesnavich model. (a) Meridional cut of the potential at $b=0$ (cylindrically symmetric, identical to the 2-DoF landscape), in the half-plane containing the symmetry axis; the wells sit on the axis at $r=r_e$ (markers), the equatorial bending ridge is the raised region off-axis, and the flat roaming shelf lies beyond, with the dashed curve the $E=0.5\ \mathrm{kcal\,mol^{-1}}$ equipotential. The closed loops are the computed FR1 orbit (the action-$13.755$ branch) at $b=0,0.3,0.57$ (Section~\ref{sec:results}), whose invariant plane is this meridional plane; the loop deforms, breaking its up--down symmetry, as the coupling increases. The star marks the $\mathrm{CH_3^+}$ center of mass. (b) Radial cuts: on the symmetry axis ($\theta=0$, the well, with no bend) and at the equator ($\theta=\pi/2$) for two azimuths, $\phi=0$ and $\phi=\pi/3$, whose separation is the three-fold modulation; the shaded band is the radial range of FR1 and the dotted line the energy $E=0.5\ \mathrm{kcal\,mol^{-1}}$. (c) The angular landscape on the sphere of directions $(\theta,\phi)$ at a representative roaming radius $r=3.1\ \AAA$ for $b=0.8$: the equatorial bending ridge carries a three-fold azimuthal corrugation, with the favored channels (triangles) near $\phi=\tfrac{\pi}{3},\pi,\tfrac{5\pi}{3}$.}
\label{fig:landscape}
\end{figure}

With the landscape established, the meaning of roaming in the model is fixed: motion at energies in the roaming window that stays on the shelf, wandering in $(\theta,\phi)$ at large $r$, rather than crossing inward to the well or outward to dissociation. The phase-space structures that separate these outcomes, and the effect of the corrugation on them, are the subject of the rest of the paper.

\subsection{The phase-space structures that organize roaming}\label{sec:po}
Roaming is a property of trajectories: a trajectory roams or does not according to
how it crosses three dividing surfaces, defined below. Each dividing surface is
anchored on an invariant manifold of the flow. In the 2-DoF model these manifolds
are the three periodic orbits of Maugui\`ere, Collins, Ezra, Farantos, and Wiggins \cite{Mauguiere2014}; in the
3-DoF model they are three-dimensional manifolds, constructed in
Section~\ref{sec:nhim}. This section reviews the 2-DoF structures, states the
trajectory classification, and gives the dimension count that determines what the
corresponding 3-DoF objects must be.

Dimension counts are taken in the energy surface $\Sigma_E$ unless stated otherwise;
all dividing surfaces in this paper are isoenergetic, constructed within a single
$\Sigma_E$, since the conserved energy confines every trajectory to its own energy
surface and energy enters the construction only as a parameter. For $n$ degrees of
freedom the phase space is $2n$-dimensional and the energy surface
$\Sigma_E=\{H=E\}$ is $(2n-1)$-dimensional. A dividing surface must locally separate
$\Sigma_E$ into two components, so its dimension is $2n-2$: codimension one in
$\Sigma_E$. The dividing surfaces of phase-space transition state theory are
anchored on a normally hyperbolic invariant manifold (NHIM) of dimension $2n-3$,
codimension two in $\Sigma_E$, whose stable and unstable manifolds have dimension
$2n-2$ and channel trajectories through the bottleneck \cite{WaalkensSchubertWiggins2008,Wiggins2016}. These counts follow from the separation requirement alone and hold for
every $n$; the dimensions of the objects satisfying them grow with $n$
(Table~\ref{tab:dims}).

\begin{table}[h]
\centering
\begin{tabular}{lcc}
\hline
 & 2 DoF (energy surface $3$-D) & 3 DoF (energy surface $5$-D) \\
\hline
dividing surface \quad(codim 1) & $2$-D & $4$-D \\
anchoring NHIM \quad(codim 2) & $1$-D: \emph{a periodic orbit} & $3$-D: not an orbit \\
its $W^s,\,W^u$ \quad(codim 1) & $2$-D & $4$-D \\
\hline
a periodic orbit & \emph{is} the NHIM & $1$-D (codim 4) \\
\quad its $W^s,\,W^u$ & (the codim-1 separatrices) & $3$-D (codim 2): cannot divide \\
\hline
\end{tabular}
\caption{The codimension of the dividing surface and of its anchoring NHIM is the same in two and three degrees of freedom; the dimension is not. In a three-dimensional energy surface the codimension-two NHIM is one-dimensional, a periodic orbit. In a five-dimensional energy surface it is three-dimensional; a periodic orbit there has codimension four, its stable and unstable manifolds codimension two, and neither separates the energy surface.}
\label{tab:dims}
\end{table}

\paragraph{Two degrees of freedom.}
For $n=2$ the energy surface is three-dimensional and a codimension-two NHIM is
one-dimensional: a periodic orbit. This is why periodic orbits organize the 2-DoF
theory. The model has three, named by their roles \cite{Mauguiere2014,Wiggins2026}.

The linear stability of a periodic orbit is measured by its Floquet multipliers~\cite{MeyerHallOffin2009}, the
eigenvalues of the monodromy matrix --- the linearization of the flow over one
period (Appendix~\ref{app:po}). The orbit is hyperbolic when a reciprocal pair
$\{\lambda,\lambda^{-1}\}$ lies off the unit circle; $\lambda>1$ is the factor by
which a transverse perturbation grows over one period. We quote $\lambda$ for each
orbit because these rates recur in Section~\ref{sec:nhim}: the normal expansion rate
of each three-dimensional manifold is the $\lambda$ of its generating orbits, and
the persistence of the manifolds for $b>0$ (Proposition~\ref{prop:persist}) requires these rates
to dominate the tangential growth.

\emph{OTS-PO, the entrance.} A relative equilibrium at the centrifugal barrier
($r_0\approx13.4$~\AAA; Fig.~\ref{fig:threeorbits}b), with $\lambda_{\mathrm{OTS}}=
1.13$ --- weakly hyperbolic because the barrier is shallow. Its dividing surface
gates the entrance: a trajectory crossing outward dissociates to
$\mathrm{CH_3^+}+\mathrm{H}$.

\emph{FR1, the classifier.} The free-rotor orbit on the roaming shelf
($r\approx3.18$--$3.65$~\AAA; Fig.~\ref{fig:threeorbits}a), associated with a 2:1
stretch--bend resonance born in a center--saddle bifurcation \cite{Mauguiere2014};
$\lambda_{\mathrm{FR1}}=24.57$. Its dividing surface does
not separate reactants from products; it is the surface whose crossings are counted.

\emph{TTS-PO, the reactive gate.} A libration through the symmetry axis, arcing
across the mouth of the well just outside the ridge crest
($r\approx2.38$--$2.64$~\AAA; Fig.~\ref{fig:threeorbits}a);
$\lambda_{\mathrm{TTS}}=220.5$, the most unstable of the three. Its dividing surface
gates capture into the well.

Trajectories are classified against these three surfaces
\cite{Mauguiere2014,Mauguiere2016}. Initial conditions are sampled on the
inward-crossing half of the OTS surface ($p_r<0$) and integrated until they cross
the TTS surface (reactive) or recross the OTS surface outward (non-reactive);
crossings of the FR1 surface are counted along the way. A trajectory roams when it
crosses the FR1 surface at least three times (reactive; the count is odd) or at
least four times (non-reactive; even), and is direct otherwise. The classification
refers only to the three surfaces; the orbits enter as their anchors. All
computations in this paper are at $E=0.5$~kcal\,mol$^{-1}$ above the dissociation
threshold, the regime of the 2-DoF studies \cite{Mauguiere2014,Wiggins2026}; the one exception is
Section~\ref{sec:fractions}, where the energy dependence of the corrugation effect is computed,
the corrugation depth $b\,V_0$ being comparable to the available energy only near
threshold.

\paragraph{Three degrees of freedom.}
For $n=3$ the energy surface is five-dimensional and a codimension-two NHIM is
three-dimensional: not a periodic orbit. A periodic orbit is one-dimensional and
therefore has codimension four in $\Sigma_E$; its stable and unstable manifolds have
dimension at most three, codimension two, and cannot separate $\Sigma_E$
(Table~\ref{tab:dims}). The three orbits above therefore cannot anchor dividing
surfaces in the 3-DoF model.

Section~\ref{sec:nhim} constructs the three-dimensional NHIMs from the $b=0$
symmetry. At $b=0$ the azimuth $\phi$ is cyclic and $p_\phi$ is conserved; at each
fixed $p_\phi$ the reduced system is a two-degree-of-freedom system carrying reduced
counterparts of the three orbits. The union of these reduced orbits over
$p_\phi\in(-p_*,p_*)$, each carrying its azimuthal circle, is a three-dimensional
invariant manifold, and it --- not the periodic orbit --- anchors the
four-dimensional dividing surface. The planar orbit is recovered in the limit
$p_\phi\to0$; the precise sense in which it is a distinguished member of its family
is treated in Section~\ref{sec:separatrix}. The trajectory classification is
unchanged: the same three surfaces and the same crossing counts, with initial
conditions on the OTS surface now sampled over $p_\phi$ as well as over the incoming
directions.

What remains computable from a periodic orbit directly is its linear stability
transverse to the plane of motion: whether the coupling $b$ makes the out-of-plane
direction unstable. That computation occupies Sections~4--5. FR1 is located on the
reaction plane $\{Y=0,\,p_Y=0\}$ (Section~\ref{sec:model}(iii)) by symmetric shooting and continued in $b$ and $E$
(Appendix~\ref{app:po}); at $b=0$ it reduces to the 2-DoF orbit, $X_0=3.179$~\AAA,
period $T=5.929$, action $13.755$.

\begin{figure}[tbp]
\centering
\includegraphics[width=0.95\textwidth]{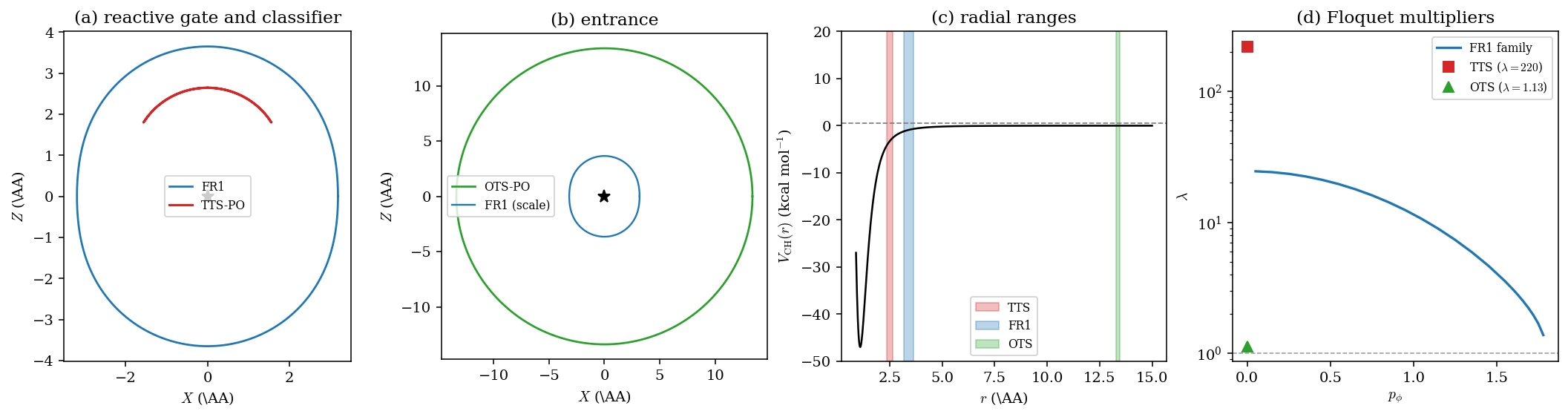}
\caption{The three periodic orbits of the 2-DoF model at $b=0$, $p_\phi=0$, $E=0.5$~kcal\,mol$^{-1}$. In three degrees of freedom each is the planar limit ($p_\phi\to0$) of a three-dimensional NHIM (Section~\ref{sec:nhim}) and does not itself anchor a dividing surface. (a) TTS-PO (red), librating through the symmetry axis and arcing across the mouth of the well just outside the ridge crest, $r\approx2.38$--$2.64$~\AAA, and FR1 (blue), $r\approx3.18$--$3.65$~\AAA; the marker at the origin is the $\mathrm{CH_3^+}$ center of mass. (b) OTS-PO (green) at the centrifugal barrier, $r_0\approx13.4$~\AAA; FR1 shown for scale. (c) On-axis potential $V_{\mathrm{CH}}(r)$ with the radial range of each orbit shaded; $E=0.5$~kcal\,mol$^{-1}$ dashed. (d) Floquet multiplier $\lambda$ (defined in Section~\ref{sec:po}): the FR1 family over $p_\phi$, decreasing from $\lambda=24.57$ at $p_\phi=0$ (Section~\ref{sec:nhim}), and, on the logarithmic scale, the values for TTS-PO ($\lambda=220.5$) and OTS-PO ($\lambda=1.13$) at $p_\phi=0$.}
\label{fig:threeorbits}
\end{figure}

\subsection{Linear stability and the transverse direction}\label{sec:stability}

Whether the third degree of freedom is dynamically active is decided by the linear stability of the orbit FR1 in the out-of-plane direction --- the direction opened up by that degree of freedom and absent from the 2-DoF problem. The monodromy matrix of the orbit, obtained by integrating the variational equations over one period, block-decomposes because the orbit lies in the reaction plane: an in-plane block (the coordinates $r,\theta$ and their conjugates, inherited from the 2-DoF dynamics) and a transverse block (the out-of-plane pair $Y,p_Y$). The transverse block is the new content. The trace of the transverse block decides whether a trajectory nudged off the reaction plane is restored to it --- elliptic motion, the perturbation bounded and quasiperiodic, the out-of-plane direction inert --- or departs from it: hyperbolic motion, the perturbation growing exponentially, the out-of-plane direction an escape direction. The value $-2$ marks the transverse period-doubling between them.

\section{Results: phase-space geometry versus coupling and energy}\label{sec:results}

\subsection{The FR1 orbit and its action}

At $E=0.5$~kcal\,mol$^{-1}$ and $b=0$ the planar orbit coincides with the 2-DoF
orbit by the reduction of Section~\ref{sec:model}: the computed abbreviated action
$W=\oint p\cdot dq$, evaluated along the orbit as $\oint 2T\,dt$, is $13.755$,
against the 2-DoF value $13.755$, with matching period and turning points. The
action also discriminates between the two branches admitted by the symmetric
shooting: the correct orbit and a spurious wide branch of smaller action
($\approx12.6$) onto which over-large continuation steps can jump
(Appendix~\ref{app:po}). Continued in $b$ within the reaction plane, the orbit deforms
smoothly: by $b\approx0.57$ the inner turning point has moved from $3.18$ to
$2.67$~\AAA, the outer staying near $3.65$~\AAA, and the action rises from $13.8$ to
$14.4$, about four percent. The reaction plane is invariant for every $b$
(Section~\ref{sec:model}(iii)), and within it the dynamics is a
two-degree-of-freedom system in which FR1 anchors a periodic-orbit dividing surface
whose flux is this action. The four-percent change states that the coupling barely
alters the planar subsystem; its decisive effect is transverse to the plane, to
which we now turn.

\subsection{The transverse period-doubling bifurcation}

The transverse trace (Fig.~\ref{fig:trace}) equals $+2$ exactly at $b=0$, as
conservation of $p_\phi$ requires, and rises above $+2$ immediately: the
out-of-plane direction is hyperbolic for all $b>0$ outside the narrow elliptic
window $0.58\lesssim b\lesssim0.63$, which is terminated by a transverse
period-doubling at $b_c\approx0.63$. Arbitrarily weak coupling produces positive transverse hyperbolicity, so the
initial instability has no finite threshold; after the narrow elliptic
restabilization window, the crossing of $\mathrm{tr}\,M_\perp=-2$ at $b_c$ marks
the onset of negative hyperbolicity through a period-doubling bifurcation. With the spherical ratio $I_z=I_x$ the small-$b$
behavior is instead elliptic: the transverse character is fixed jointly by the
$C_3$ coupling and the oblate inertia $I_z=2I_x$, and neither alone suffices.

The four multipliers normal to the flow within the energy surface make this quantitative: at $b=0.30$ the hyperbolic pair of the planar orbit is $\{24.6,\ 0.041\}$ --- the normal rate that Section~\ref{sec:nhim} inherits --- and the transverse pair is $\{2.40,\ 0.42\}$, quantifying the instability.

Hyperbolic or not, the orbit is one-dimensional; its stable and unstable manifolds
are three-dimensional, codimension two in the five-dimensional energy surface, and
do not separate it. What the transverse instability supplies is the ingredient the
planar orbit lacks --- a genuine out-of-plane expanding direction. How that
direction relates to the three-dimensional manifold of the FR1 transition state is
answered by constructing the manifold (Section~\ref{sec:nhim}); the relationship
between the orbit and its family is made precise in Section~\ref{sec:separatrix}.

Figure~\ref{fig:excursion} in Section~\ref{sec:dimensionality} shows the consequence directly: a trajectory displaced out of the reaction plane has a bounded out-of-plane excursion at $b=0$ and a geometrically growing one for representative $b>0$ outside the elliptic window.

\begin{figure}[t]
\centering
\includegraphics[width=0.82\textwidth]{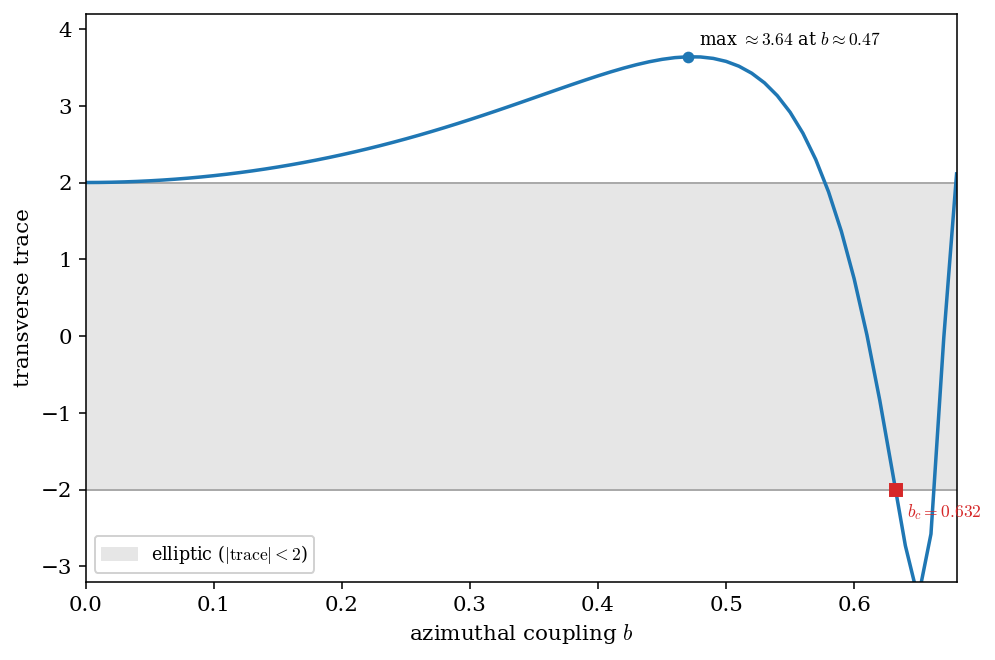}
\caption{Trace of the transverse block of the monodromy matrix of the FR1 orbit versus the azimuthal coupling $b$, at $E=0.5\ \mathrm{kcal\,mol^{-1}}$ with the physical inertia ratio $I_z=2I_x$ and the methyl ($\cos3\phi$) coupling of Eq.~\eqref{eq:Vcoup}. The shaded band $|\mathrm{trace}|<2$ is the elliptic (transverse-stable) regime. The trace rises above $+2$ immediately above $b=0$ --- the out-of-plane direction is hyperbolic, and the orbit a NHIM --- reaches its maximum of $\approx3.6$ near $b\approx0.45$, turns elliptic in the narrow window near $b\approx0.6$, and undergoes a transverse period-doubling (trace $=-2$) at $b_c\approx0.63$.}
\label{fig:trace}
\end{figure}

\section{Dimensionality of roaming: the formaldehyde--acetaldehyde contrast}\label{sec:dimensionality}

Whether roaming requires a three-dimensional description, or a planar model captures
it, distinguishes the two systems in which roaming has been studied since its
identification: formaldehyde and acetaldehyde \cite{Townsend2004,KrajnakWiggins2024}. Setting the present model
against that contrast motivates its construction and fixes the chemical reading of
the transverse instability found above.

In formaldehyde, a reduced two-degree-of-freedom phase-space model --- the roaming H atom moving in a plane relative to a rigid HCO fragment --- reproduces the roaming mechanism in $\mathrm{H_2CO}$ decomposition \cite{Townsend2004,Mauguiere2015}, and the reduction has been validated against full-dimensional quasiclassical trajectory studies \cite{Houston2016}. The reduction is faithful for a dynamical reason rather than by accident: at the large fragment separations where roaming occurs the relevant angular momentum is nearly conserved on an essentially fixed plane, the angular and radial motions decouple, and a centrifugal barrier provides the gatekeeping. Formaldehyde is thus the prototype in which roaming is in-plane and centrifugal-barrier-mediated, and in which the out-of-plane direction is not where the mechanism resides.

In acetaldehyde the third degree of freedom is essential. Roaming in the photodissociation of $\mathrm{CH_3CHO}$ to $\mathrm{CH_4}+\mathrm{CO}$ was identified experimentally by Houston and Kable, from CO product-state distributions the conventional transition state cannot account for~\cite{HoustonKable2006}, and corroborated by slice-imaging measurements~\cite{RubioLago2007}; combined experiment and full-dimensional quasiclassical trajectories on a global potential energy surface then established roaming as the dominant channel to the molecular products, proceeding by abstraction of an H atom from HCO by the methyl group~\cite{Heazlewood2008,SheplerBraamsBowman2008}. Here the roaming partner is a methyl group rather than a hydrogen atom, and full-dimensional trajectory studies, supplemented by a restricted two-degree-of-freedom model, reveal two disjoint roaming pathways \cite{KrajnakWiggins2024}: a long-range channel, with maximum $\mathrm{CH_3}$--HCO separations of roughly $14.5$--$22.9\ \mathrm{a.u.}$, that the in-plane model reproduces and that proceeds, formaldehyde-like, through a centrifugal barrier; and a short-range channel, near $9$--$11.5\ \mathrm{a.u.}$, that the in-plane model cannot produce and in which the fragment undergoes substantially more rotation --- the signature of active out-of-plane motion. That the heavier roaming fragment is not by itself the explanation is established within Chesnavich's $\mathrm{CH_4^+}$ model itself: varying the mass of the roaming fragment does not significantly change the roaming propensity \cite{KrajnakWiggins2018}. What distinguishes acetaldehyde is therefore not mass but the presence of a second, out-of-plane pathway absent from the planar description.

The one-parameter family $H_b$ interpolates between these two situations along precisely the axis the contrast identifies. At $b=0$ the model is cylindrically symmetric, $p_\phi$ is conserved, and the azimuthal degree of freedom is decoupled, each trajectory confined to its $p_\phi$-leaf --- the analog of the in-plane, angular-momentum-conserving regime in which formaldehyde is effectively two-dimensional. Increasing $b$ resolves the three-fold azimuthal structure of the methyl frame and couples the out-of-plane motion in --- the analog of activating the out-of-plane participation that acetaldehyde displays. The spatial structure of the coupling reinforces the parallel. The corrugation it introduces is strongest at the inner edge of the roaming shelf, where $V_0$ remains of order the roaming energy, and fades as $V_0\to0$ at large $r$ (Section~\ref{sec:landscape}); the model therefore carries its out-of-plane structure at short range and leaves the long-range motion in-plane --- the same short-range/long-range division of labour that separates the two acetaldehyde channels.

What the model contributes beyond the empirical contrast is a direct, quantitative measure --- computed within a single tractable system --- of when the third degree of freedom becomes dynamically active. We probe it with trajectories launched from initial conditions displaced a small distance $\delta$ out of the reaction plane from a point of the FR1 orbit, following the out-of-plane excursion $|Y(t)|$ (Fig.~\ref{fig:excursion}). At $b=0$ the excursion remains bounded, oscillating at fixed amplitude: the out-of-plane direction is marginal and a trajectory nudged out of the plane is neither restored nor expelled, as the conserved $p_\phi$ requires. Once the corrugation is switched on the excursion grows, and it grows faster as the coupling increases --- already by an order of magnitude over a few roaming periods at $b=0.15$, and far more steeply by $b=0.45$. The growth is geometric: the excursion grows by the transverse Floquet multiplier of the orbit each period, and this multiplier rises from unity at $b=0$ to $\lambda\approx3.6$ near $b\approx0.45$ (Section~\ref{sec:results}). Switching on the corrugation thus converts the marginal out-of-plane direction into an unstable one. The single bounded curve at $b=0.60$ in Fig.~\ref{fig:excursion} is the exception that the transverse-stability analysis predicts --- the narrow window of transverse re-stabilization just below the period-doubling at $b_c$ (Section~\ref{sec:results}) --- and it does not alter the dominant trend, that the out-of-plane direction is unstable across most of the coupling range.

The computation locates the effect of the corrugation. It is not on the in-plane
reactive flux: the FR1 action rises about four percent by $b\approx0.57$ and under
seven percent up to $b_c$, and the in-plane instability that governs residence on
the shelf (multiplier of order twenty) is essentially unchanged. The effect is the
opening of an out-of-plane component to the roaming motion. The third degree of
freedom does not confine trajectories to the reaction plane; outside the narrow
elliptic window it is transversely unstable, an additional route by which a roaming
trajectory leaves the plane, with no finite threshold in $b$
(Section~\ref{sec:results}). The model makes concrete what the
formaldehyde--acetaldehyde comparison shows empirically: resolving the discrete
azimuthal structure of the roaming partner activates an out-of-plane degree of
freedom a planar model cannot represent, and it is this activation --- not a change
in the in-plane bottleneck --- that separates three-dimensional roaming from its
planar reduction.

The correspondence is structural, not quantitative. The model is the
$\mathrm{CH_4^+}$ system, not formaldehyde or acetaldehyde; the coupling strength
$b$ and the bifurcation value $b_c$ have no counterpart in any particular molecule.
And the physical origin of the out-of-plane activation differs: in acetaldehyde the
short-range channel is attributed to short-range repulsive structure absent from the
reduced model, whereas here the coupling arises from the $C_3$ corrugation of the
attractive bending ridge. What the systems share is the role of the methyl group's
discrete symmetry and the participation of out-of-plane motion at short range.

\begin{figure}[t]
\centering
\includegraphics[width=0.82\textwidth]{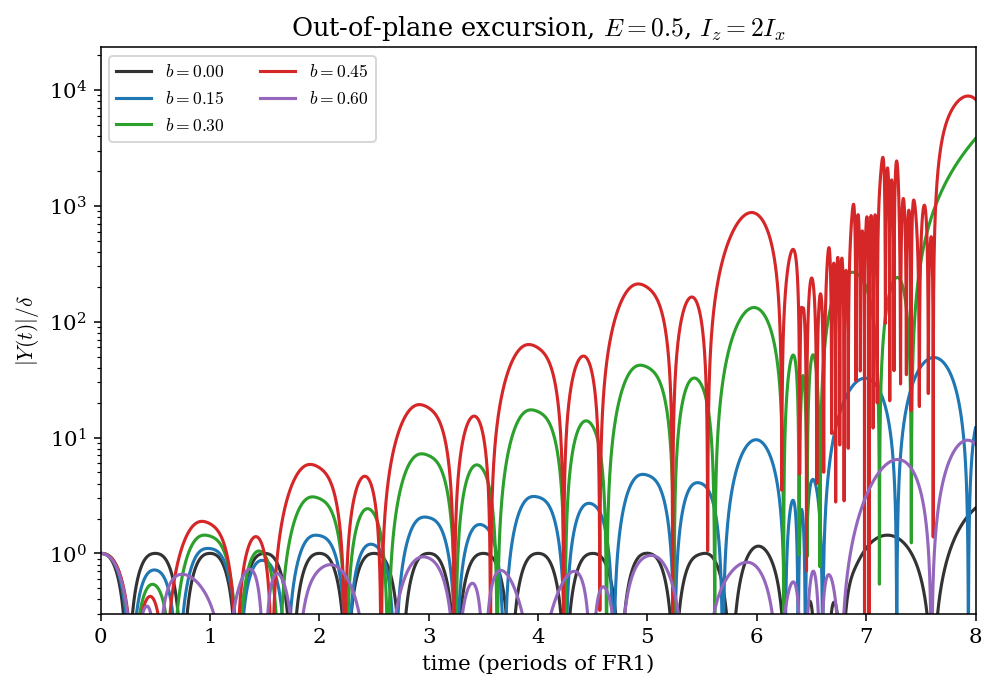}
\caption{The model's internal measure of the dimensional transition. A trajectory is launched from an initial condition displaced a small distance $\delta$ out of the reaction plane from a point of the FR1 orbit, and the out-of-plane excursion $|Y(t)|$ is followed over eight periods, at $E=0.5$~kcal\,mol$^{-1}$ with the physical inertia ratio $I_z=2I_x$. At $b=0$ (cylindrically symmetric, $p_\phi$ conserved) the excursion is bounded: the out-of-plane direction is marginal and the dynamics is effectively two-dimensional. For $b>0$ the excursion grows geometrically and more steeply with increasing coupling ($b=0.15,0.30,0.45$), the trajectory leaving the plane geometrically, by the transverse Floquet multiplier each period. The bounded curve at $b=0.60$ lies in the narrow window of transverse re-stabilization just below the period-doubling at $b_c$ (cf.\ Fig.~\ref{fig:trace}).}
\label{fig:excursion}
\end{figure}

\section{The NHIMs and their dividing surfaces: explicit construction at $b=0$ and persistence}\label{sec:nhim}

Section~\ref{sec:po} established what each of the three transition states must be in
three degrees of freedom: a four-dimensional dividing surface, anchored on a
three-dimensional NHIM whose four-dimensional stable and unstable manifolds
partition the energy surface \cite{Wiggins2016,KrajnakGarciaGarridoWiggins2021}. A periodic orbit ---
one-dimensional, with two-dimensional stable and unstable manifolds --- is not that
object. This section constructs the three NHIMs and the dividing surfaces they
anchor, and proves that every compact interior piece of each survives
sufficiently small coupling.

The construction rests on the $b=0$ symmetry. At $b=0$ the azimuth $\phi$ is cyclic
and $p_\phi$ is conserved; fixing $p_\phi$ and discarding $\phi$ reduces the
dynamics to two degrees of freedom (Section~\ref{sec:reduction}). For each fixed
$p_\phi$ the reduced system carries three unstable periodic orbits, the counterparts
at that angular momentum of the three orbits of Section~3. Collecting the orbits of
one type over the admissible range of $p_\phi$, and restoring the azimuthal angle
that the reduction removed, yields a three-dimensional invariant manifold --- one
for each transition state (Section~\ref{sec:families}).
Section~\ref{sec:normalhyp} verifies that each manifold is normally hyperbolic.
Section~\ref{sec:dstopo} constructs the dividing surface each manifold anchors and
determines its topology, following at each fixed $p_\phi$ the construction of
Maugui\`ere, Collins, Kramer, Carpenter, Ezra, Farantos, and Wiggins for ozone \cite{Mauguiere2016}; what is taken from that work and
what the present setting adds is stated there. Section~\ref{sec:separatrix} settles
the relation between the planar orbit FR1 and the manifold it generates: the orbit
is a distinguished member of its family, in a sense made precise there, and the
marginal transverse stability found in Section~\ref{sec:results} is the signature of
that status. Section~\ref{sec:persist} proves persistence of the compact interior
pieces for sufficiently small $b>0$; the trajectory study of Section~7
stands on its own measurements, the $b=0$ construction supplying the
surface placement it uses (Section~\ref{sec:robust}).

We carry out each step for FR1 in detail; the OTS and TTS manifolds follow by the
same construction, and the points at which they differ are recorded where they
occur.

\subsection{Symmetry reduction at $b=0$}\label{sec:reduction}

At $b=0$ the coupling \eqref{eq:Vcoup} reduces to $V_0(r)\sin^2\theta$, the
potential depends on $(r,\theta)$ alone, the azimuth $\phi$ is cyclic, and the
azimuthal angular momentum $p_\phi=L_z$ is conserved. Conservation of $p_\phi$ is
exact, and we use it as a check on the numerical integration and on the coded
equations of motion: on a generic non-planar, non-axial trajectory the integrator
holds $p_\phi$ constant to $1\times10^{-11}$, and the energy to $3\times10^{-10}$,
over an integration time of $60$. Fixing $p_\phi$ and discarding the cyclic angle
$\phi$ yields the two-degree-of-freedom reduced Hamiltonian, obtained directly from
\eqref{eq:H}:

\begin{equation}\label{eq:Hred}
H_{\mathrm{red}}(r,\theta,p_r,p_\theta;p_\phi) =
\frac{1}{2I_x}\left(p_\theta^2+p_\phi^2\cot^2\theta\right)
+\frac{p_\phi^2}{2I_z}
+\frac{1}{2m}\left(p_r^2+\frac{p_\theta^2}{r^2}
+\frac{p_\phi^2}{r^2\sin^2\theta}\right)
+V_{\mathrm{CH}}(r)+V_0(r)\sin^2\theta.
\end{equation}

Collecting the terms in $p_\phi^2$, the angular momentum enters only through the
effective centrifugal potential
$U(r,\theta;p_\phi)=p_\phi^2\!\left[\frac{1}{2mr^2\sin^2\theta}
+\frac{\cot^2\theta}{2I_x}+\frac{1}{2I_z}\right]$,
which diverges at the poles $\theta=0,\pi$ for $p_\phi\neq0$, vanishes identically
at $p_\phi=0$, and contains the only appearance of $I_z$
(Section~\ref{sec:inertia}).

At $p_\phi=0$ equation \eqref{eq:Hred} is the two-degree-of-freedom Chesnavich model
of Section~\ref{sec:model} \cite{Wiggins2026,Mauguiere2014}, whose orbits are those of
Sections~3--5.

\subsection{The families of reduced orbits and the two-component
manifolds}\label{sec:families}
For each fixed $p_\phi$ the reduced system \eqref{eq:Hred} possesses, in each of the
three regions, an unstable periodic orbit --- the counterpart, at that angular
momentum, of the corresponding planar orbit of Section~3. We denote the FR1-type
reduced orbit by $\gamma_{p_\phi}$, reserving the names FR1, OTS-PO, TTS-PO for the
planar orbits at $p_\phi=0$; we treat the FR1 family explicitly and return to the
other two below. Each orbit is located by symmetric shooting
(Appendix~\ref{app:po}). For $p_\phi\neq0$ it is a \emph{relative} periodic orbit:
it closes in the reduced variables $(r,\theta,p_r,p_\theta)$ after the period
$T(p_\phi)$ of the reduced orbit, while the azimuth advances, so the full-space
motion need not close. Because \eqref{eq:Hred} depends on $p_\phi$ only through
$p_\phi^2$, the reduced orbits at $\pm p_\phi$ coincide and differ only in the sense
of the azimuthal drift: they form a counter-precessing pair. (These components are
the analog of the counter-propagating pair of \cite{Mauguiere2016}, though their
role here differs --- Section~\ref{sec:dstopo}.)

Continuation in $p_\phi$ generates a one-parameter family, every member computed at
the same total energy $H_{\mathrm{red}}=E$ (the isoenergetic convention of
Section~\ref{sec:po}). As $|p_\phi|$ increases, the centrifugal potential
$U(r,\theta;p_\phi)$ of Section~\ref{sec:reduction} diverges at the poles
$\theta=0,\pi$ and confines the orbit away from the axis. Let $\theta_-(p_\phi)$
denote the smallest polar angle attained on the orbit --- the polar turning angle,
at which $p_\theta=0$; by the reflection symmetry $\theta\to\pi-\theta$ the orbit
occupies $[\theta_-,\pi-\theta_-]$. The turning angle increases monotonically with
$|p_\phi|$: for the FR1 family from $0$ at $p_\phi=0$ to roughly $44^\circ$ at
$p_\phi=1.7$, the radial extent remaining on the roaming shelf,
$r\in[3.18,3.65]$~\AAA, throughout (Fig.~\ref{fig:nhim}a).

Lifting a member back to the full phase space restores the azimuthal circle
$\phi\in S^1$ that the reduction removed. For $p_\phi\neq0$ the orbit remains at a
distance $\rho_{\min}=r_t\sin\theta_->0$ from the axis, so under the rotation each
point of the orbit sweeps a circle of positive radius, and the lift is a two-torus,
\begin{equation}\label{eq:torus}
\mathbb{T}(p_\phi)=\bigcup_{\phi\in S^1}\big\{\text{reduced orbit at }p_\phi,\ \text{rotated by }\phi\big\},
\end{equation}
the product of two circles: the reduced orbit itself, a closed
curve parametrized by the time along it, and the azimuthal circle. The torus is
invariant: the reduced orbit is invariant under the reduced flow, the azimuthal
rotation is a symmetry at $b=0$, and $p_\phi$ is conserved. The candidate NHIM ---
the manifold on which the dividing surface will be anchored --- is the union of
these tori over the angular momentum,
\begin{equation}\label{eq:M0}
\mathcal{M}_0=\bigcup_{0<|p_\phi|<p_*}\mathbb{T}(p_\phi),
\end{equation}
Every member is computed at the
same total energy and the lift leaves $H$ unchanged, so $\mathcal{M}_0$ lies in the
single five-dimensional energy surface $\Sigma_E$; it is three-dimensional, exactly
the $2n-3=3$ required for $n=3$ (Section~\ref{sec:po}).

The union is taken over $0<|p_\phi|<p_*$ for two reasons. At $|p_\phi|=p_*$ the
reduced orbit loses hyperbolicity: its Floquet multiplier $\lambda(p_\phi)$,
computed along the family in Section~\ref{sec:normalhyp}, decreases monotonically
from $24.57$ at $p_\phi=0$ to $1$ at $p_*=1.799$ (Fig.~\ref{fig:threeorbits}d), and
normal hyperbolicity fails there. At $p_\phi=0$ the construction degenerates: the
centrifugal potential vanishes identically, the orbit reaches the poles, and the
azimuthal circle shrinks to a point on the axis.

Consequently $\mathcal{M}_0$ has two disjoint components,
\[
\mathcal{M}_0^{+}=\{q\in\mathcal{M}_0:p_\phi(q)>0\},\qquad
\mathcal{M}_0^{-}=\{q\in\mathcal{M}_0:p_\phi(q)<0\},
\]
the counter-precessing components, each homeomorphic to $T^2\times(0,p_*)$. The argument
is direct: $p_\phi$ is continuous on $\mathcal{M}_0$ and never zero there, so a path
in $\mathcal{M}_0$ from $\mathcal{M}_0^+$ to $\mathcal{M}_0^-$ would carry $p_\phi$
continuously from a positive to a negative value and would have to pass through
$p_\phi=0$, which $\mathcal{M}_0$ excludes. The two components are exchanged by the
map $(\phi,p_\phi)\mapsto(-\phi,-p_\phi)$, a symmetry of $H$ --- $p_\phi$ enters
only as $p_\phi^2$, and the coupling is even in $\phi$ --- which carries one
component onto the other. Parametrising each reduced curve $\gamma_{p_\phi}$ by a
phase $s\in S^1$, the time along the orbit modulo $T(p_\phi)$, each component is, in
the variables $(s,p_\phi)$, the open annulus $S^1\times(0,p_*)$; the word
\emph{annular} below refers to one component in this sense.

The exclusion of $p_\phi=0$ is a genuine singularity of the family, not a coordinate
artifact. Near the axis the centrifugal potential grows as
$p_\phi^2/(2mr^2\theta^2)$, so the turning angle at which it balances the available
energy scales linearly in the momentum,
$\theta_-(p_\phi)=k|p_\phi|+O(p_\phi^2)$; Proposition~\ref{prop:sep} makes this
precise. The closest approach to the axis is then
$\rho_{\min}=r_t\sin\theta_-\simeq r_t k|p_\phi|\simeq1.69\,|p_\phi|$, with the
outer radial turning point $r_t\approx3.65$~\AAA\ nearly constant across the family.
In body-frame Cartesian coordinates the locus the family traces near a polar turning
point is therefore the double cone $X^2+Y^2\simeq(1.69\,p_\phi)^2$, whose apex, on
the axis at $p_\phi=0$, is a conical singularity; a quadratic law
$\theta_-\propto p_\phi^2$ would give a smooth tangency instead. The two nappes of
the cone are the two components, meeting only at the excluded apex.

The structure just described is a \emph{fibration}, and we fix the vocabulary here
because it is used throughout Sections~6.3--6.6. For a function $f$ and a value $c$,
the \emph{preimage} $f^{-1}(c)$ is the set of points at which $f$ takes the value
$c$. On $\mathcal{M}_0$ the function is the conserved momentum $p_\phi$, and the
preimage of a single value is the single torus $\mathbb{T}(p_\phi)$: the
\emph{fiber} over that value. The essential property of a fiber is that it is a
regular, interchangeable member of its family: near any value of $p_\phi$ the
manifold is a product, torus~$\times$~interval, and varying $p_\phi$ deforms each
torus smoothly into its neighbors, every member alike in kind. (In the reduced
variables the same picture is the family of curves $\gamma_{p_\phi}$ over the
momentum interval --- the open annulus above --- with the curves as fibers.) The
value $p_\phi=0$ is excluded from \eqref{eq:M0}, and the natural question is whether
the planar orbit FR1 completes the family there, as the fiber over $p_\phi=0$.
Section~\ref{sec:separatrix} shows that it does not: FR1 is related to the family
not as a fiber but as a \emph{separatrix} --- not a member of the family but a
boundary between the two, and what it separates, in which variables, is made precise
there.

\subsection{Normal hyperbolicity of each component}\label{sec:normalhyp}
For a point on $\mathcal{M}_0$ the three tangent directions are the flow, the
infinitesimal azimuthal rotation, and the family direction $\partial/\partial
p_\phi$; normal hyperbolicity requires the remaining directions --- normal to
$\mathcal{M}_0$, i.e.\ transverse to it --- to expand and contract at rates that
strictly dominate any tangential growth \cite{Fenichel1971,HirschPughShub1977,Wiggins1994}. The
instrument is the monodromy matrix of each reduced relative orbit, computed in the
full six-dimensional phase space: the $6\times6$ matrix
$M_{\mathrm{RPO}}=R(-\Delta\phi)\,D\Phi_T$, where $D\Phi_T$ is the linearization of
the flow over one period $T(p_\phi)$ of the reduced orbit, obtained by integrating
the variational equations of \eqref{eq:H} along the orbit
(Appendix~\ref{app:po}), $\Delta\phi$ is the azimuthal advance accumulated
over that period, and $R(-\Delta\phi)$ is the rotation about the symmetry axis
through $-\Delta\phi$ --- the rotation that closes the relative orbit. Its
eigenvalues are
\begin{equation}\label{eq:floquet}
\{\lambda,\ \lambda^{-1},\ 1,\ 1,\ 1,\ 1\},
\end{equation}
the algebraic multiplicity of the eigenvalue $1$ being exactly four.

The count of unit eigenvalues encodes the geometry, and is worth setting out once:
each conserved quantity contributes two, in a $2\times2$ block. For the energy: the
flow direction returns to itself after one period --- an eigenvector with eigenvalue
$1$ --- and paired with it is the direction of increasing energy, along which the
orbit continues as a one-parameter family; because the period varies with $E$, this
paired direction is in general not an eigenvector but completes a $2\times2$ Jordan
block. For the angular momentum: the rotation direction is likewise an eigenvector,
the rotation commuting with the flow at $b=0$, and the family direction
$\partial/\partial p_\phi$ completes the second block, the period varying with
$p_\phi$. The two conserved quantities $H$ and $p_\phi$ thus account for the four
unit eigenvalues --- algebraic multiplicity four, geometric multiplicity two --- and
the remaining pair $\{\lambda,\lambda^{-1}\}$, reciprocal because
$M_{\mathrm{RPO}}$ is symplectic, is the normal pair. The monodromy matrix is
computed in the full six-dimensional phase space; of the four unit directions, the
energy direction is transverse to the energy surface $\Sigma_E$ and is set aside,
and normal hyperbolicity is the statement about the remaining directions within
$\Sigma_E$ (Section~\ref{sec:po}).

At $p_\phi=0$ the computation recovers the planar orbit of
Section~\ref{sec:results} --- multipliers $\{24.57,\,0.0407,\,1,1,1,1\}$, inner
turning point $X_0=3.17907$, period $T=5.92863$, action $13.7548$ against $13.755$
from the 2-DoF computation \cite{Wiggins2026}, orbit closure $3.5\times10^{-14}$. The
two computations must agree there, since at $p_\phi=0$ the reduced system is the
2-DoF model (Section~\ref{sec:reduction}); the agreement checks the six-dimensional
machinery at the one point where the answer is known independently. That the
hyperbolic pair is normal --- acting transversally to $\mathcal{M}_0$, not along it
--- is confirmed by the spectral projector onto the hyperbolic eigenspace: the flow
and rotation directions are eigenvectors of eigenvalue $1$ to residuals
$\sim10^{-14}$, and all three tangent directions have components
$\lesssim10^{-13}$ in the hyperbolic eigenspace ($\sim10^{-5}$ for the
finite-difference family direction), so the expanding and contracting directions are
transverse to $\mathcal{M}_0$. $M_{\mathrm{RPO}}$ is symplectic --- $D\Phi_T$
because the variational flow of a Hamiltonian system is symplectic,
$R(-\Delta\phi)$ because rotations are canonical transformations --- and, as
numerical checks of this, $\det M_{\mathrm{RPO}}=1$ and the eigenvalues occur in
reciprocal pairs (Appendix~\ref{app:verification}).

The hyperbolic pair persists, real and bounded away from unity, across the FR1
family: $\lambda(p_\phi)$ decreases monotonically from $24.57$ at $p_\phi=0$ to
$\approx2.1$ at $p_\phi=1.7$, reaching $1$ at $p_*=1.799$, where the orbit turns
elliptic and normal hyperbolicity is lost (Fig.~\ref{fig:nhim}b). At every interior
value of $p_\phi$ the multiplier is real and greater than one, so the family loses
hyperbolicity only at its endpoint. The tangential rate is identically zero at $b=0$
--- the tangential dynamics is a periodic flow, a rigid rotation, and the neutral
sweep of $p_\phi$ --- and in particular the out-of-plane pair $(Y,p_Y)$ lies in the
unit eigenspace, tangent to $\mathcal{M}_0$. With the tangential rate exactly zero,
the normal rates dominate any power of tangential growth, and each component of
$\mathcal{M}_0$ is normally hyperbolic on every closed sub-annulus
$\varepsilon\le|p_\phi|\le p_*-\delta$ bounded away from both ends, the normal rate
$\nu=\ln\lambda/T$ running from $0.54$ near the center to zero at $p_*$
(Fig.~\ref{fig:nhim}c).

\begin{figure}[t]
\centering
\includegraphics[width=\textwidth]{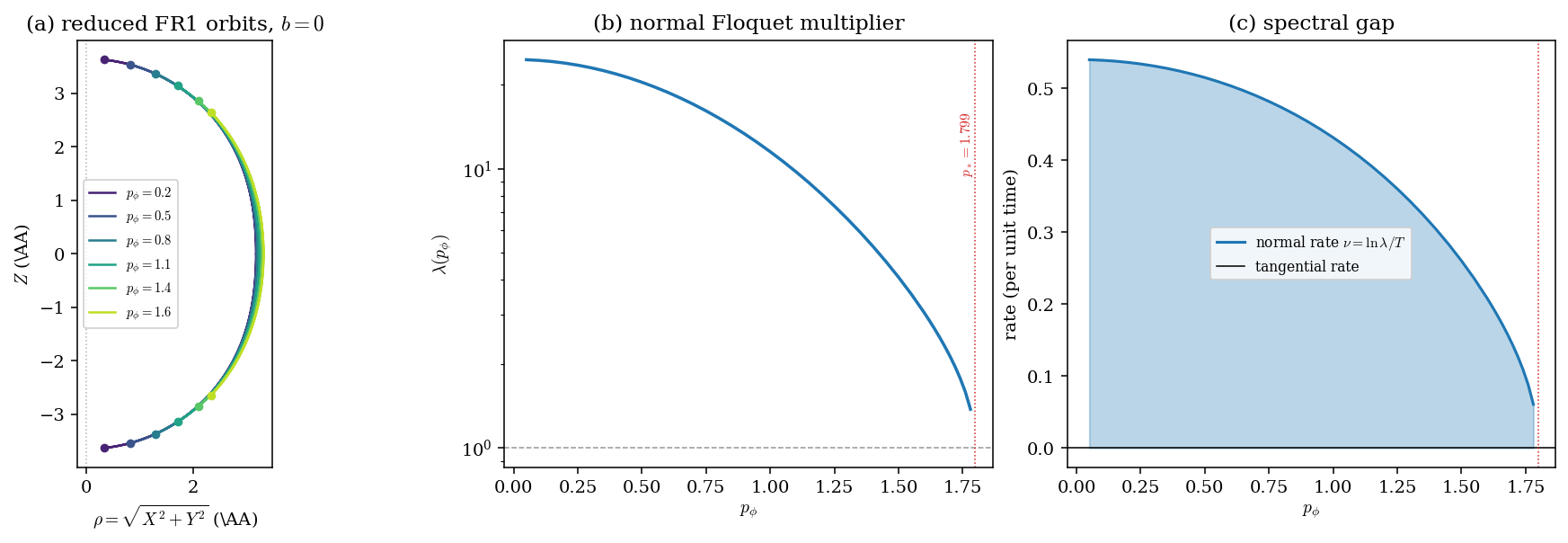}
\caption{The two-component NHIM of FR1 at $b=0$, $E=0.5$. (a) The family of reduced orbits in the meridional plane $(\rho,Z)$, $\rho=\sqrt{X^2+Y^2}$, for several values of the angular momentum $p_\phi$; markers show the polar turning points, which move away from the symmetry axis (the poles) as $p_\phi$ increases while the orbits stay on the roaming shelf. Lifting this family over the azimuthal circle and assembling over $p_\phi\in(0,p_*)$ (plus the $-p_\phi$ counterpart) produces the two-component three-dimensional manifold \eqref{eq:M0}. (b) The normal Floquet multiplier $\lambda(p_\phi)$ (logarithmic scale): real and bounded away from unity throughout, $\lambda$ real and greater than one at every interior $p_\phi$, decreasing to $\lambda=1$ at the boundary $p_*=1.799$ where the orbit turns elliptic. (c) The spectral gap: the normal Lyapunov rate $\nu=\ln\lambda/T$ (shaded) is strictly positive on the interior while the tangential rate is identically zero, so the gap is open everywhere inside $|p_\phi|<p_*$.}
\label{fig:nhim}
\end{figure}

\subsection{The dividing surface and its topology}\label{sec:dstopo}
The construction of the dividing surface at the level of the reduced
two-degree-of-freedom system is due to Maugui\`ere, Collins, Kramer,
Carpenter, Ezra, Farantos, and Wiggins, who developed it in their phase-space
analysis of roaming in ozone \cite{Mauguiere2016} --- a system of the same
structural type as the present model at $b=0$: three degrees of freedom
reducing to two at each fixed value of a decoupling parameter, there the
adiabatically conserved energy of the third mode, here the exactly conserved
$p_\phi$. Since this subsection leans on their results, we state precisely
what is taken from that work, what is adapted from it, and what does not
arise there. From \cite{Mauguiere2016} we take the complete reduced theory:
for an unstable periodic orbit of a two-degree-of-freedom system of this
kinetic form, at fixed energy, (i) the topology of the dividing surface is
decided by whether the generating orbit \emph{librates} --- oscillating
between two turning points at which its momentum ellipse degenerates to a
point --- or \emph{circulates}, never meeting such a degeneration (their
type~1 and type~2 orbits, respectively); (ii) a librating orbit yields a
dividing surface that is a two-sphere, on which the orbit is a dividing
equator separating the two hemispheres, the forward and backward crossing
directions; (iii) a circulating orbit yields instead a two-torus, and ---
the subtle point of that work --- a single circulating orbit does \emph{not}
separate its torus: the dividing surface must be assembled from the orbit
\emph{and its counter-precessing component}, and only the pair divides; and
(iv) a flux form certifies that the Hamiltonian vector field is transverse
to the surface everywhere off the generating orbit, so the surface is
locally free of recrossing. Both types occur in their ozone analysis. The
assembly over the third degree of freedom is likewise adapted from that
work: there the four-dimensional dividing surface is a union of leaves over
the adiabatic partition of energy between the reactive pair and the
decoupled vibration; here it is a union of leaves over the exactly conserved
momentum $p_\phi$, each leaf the lift of a reduced surface over the
reconstructed azimuthal circle. What does not arise there follows from the
difference in parameter: their energy partition runs over a half-line, while
$p_\phi$ runs over a signed interval, producing the two counter-precessing
components and the conical center with its separatrix orbit
(Sections~\ref{sec:families}, \ref{sec:separatrix}); and since their third
degree of freedom decouples adiabatically rather than through an exact
symmetry broken by a coupling, the persistence question of
Section~\ref{sec:persist} has no analog there.

The reduced theory applies verbatim because the reduced system
\eqref{eq:Hred} is, in form, identical to the one analyzed there: with
$q_1=r$, $q_2=\theta$, reduced mass $m$, and the inertia $I_x$, the reduced
kinetic energy $p_\theta^2(\tfrac1{2I_x}+\tfrac1{2mr^2})$ is exactly the
form treated in \cite{Mauguiere2016}, with the centrifugal terms folding
into the effective potential $U(r,\theta;p_\phi)$ of \eqref{eq:Hred}. For
every $p_\phi\neq0$ the reduced FR1 orbit librates between two
$\theta$-turning points at which \emph{both} reduced momenta vanish --- the
degeneration of the momentum ellipse that the sphere case requires,
established numerically in Appendix~\ref{app:verification} --- so its
reduced dividing surface is a two-sphere with the orbit as a dividing
equator through the two degenerate points. The flux form --- Eq.~(B12c) of
\cite{Mauguiere2016}, written in the present variables ---
\begin{equation}\label{eq:flux}
\varphi \;=\; 1-\frac{m\,p_\theta}{I_1(r)\,p_r}\,
\frac{dr}{d\theta}\bigg|_{\mathrm{PO}},\qquad
\frac1{I_1}=\frac1{I_x}+\frac1{mr^2},
\end{equation}
in which the momenta are those of the surface point under test while the
derivative $dr/d\theta$ is taken along the orbit's configuration-space
projection, vanishes precisely at the orbit's two momentum points on each
momentum ellipse --- the inbound and outbound crossings --- and is nonzero
elsewhere, so the Hamiltonian vector field is transverse to the surface off
the orbit and there is no local recrossing.

Restoring the azimuthal circle lifts this picture by one dimension
(Fig.~\ref{fig:lift}). At a fixed nonzero angular momentum $p_\phi$ every point
of the reduced dividing surface is carried around the azimuthal angle $\phi$,
which itself runs over a circle $S^1$; because the reduced surface lies entirely
off the symmetry axis, $\phi$ is a single-valued coordinate there and the lift
is a global product~\eqref{eq:torus} (Proposition~\ref{prop:topology}). The
reduced dividing surface is a two-sphere $S^2$ (the surface of a ball), so its
lift is the product $S^2\times S^1$: a sphere's worth of points, each carried
once around a circle. On that sphere the reduced orbit is an equator; carried
around the azimuthal circle it sweeps out a two-torus $\mathbb{T}(p_\phi)$ (the
surface of a doughnut), which is the slice of the NHIM. Cutting a sphere along
its equator leaves two caps, each a disk $D^2$; performing the same cut once the
azimuthal circle has been attached separates $S^2\times S^1$ along the NHIM
torus into two \emph{solid tori} (solid doughnuts $D^2\times S^1$),
\begin{equation*}
(S^2\times S^1)\setminus(\text{equator}\times S^1)
=(S^2\setminus\text{equator})\times S^1
=(D^2\sqcup D^2)\times S^1,
\end{equation*}
where $\sqcup$ denotes disjoint union. The two pieces are the forward and
backward halves of the dividing surface, each bounded by the NHIM and each
carrying one crossing direction. The full
four-dimensional dividing surface is the union of these three-dimensional
slices over each component of the manifold. Two features distinguish this
from the two-degree-of-freedom case. First, the NHIM torus at a single value
of $p_\phi$ already divides its own slice --- the reduced surface is a
sphere, and an equator separates a sphere --- so the counter-precessing component
is \emph{not} needed to obtain a dividing surface at fixed $p_\phi$, in
contrast to the circulating (type-2) case of \cite{Mauguiere2016}. Second,
the pairing is global rather than local: the $+p_\phi$ and
$-p_\phi$ components are the two disjoint components of $\mathcal{M}_0$,
meeting only through the excluded singular apex.

A reader of \cite{Mauguiere2016} might expect both types to appear among the
three transition states, as they do in ozone, where the tight orbits librate
and the orbiting orbit circulates. The realization here is different, and
the geometry forces it. For $p_\phi\neq0$ no coordinate of the reduced
system is periodic: the centrifugal barrier walls off both poles, confining
$\theta$ to an interval, and $r$ is not an angle. Circulation is therefore
impossible, and every member of all three families librates, with two
genuine turning points at which both reduced momenta vanish
(Appendix~\ref{app:verification}): the FR1- and OTS-type members across the
equator, between the polar turning angles $\theta_-(p_\phi)$ and
$\pi-\theta_-(p_\phi)$; the TTS-type members on one side of the equator,
their planar limit librating through the axis with
$\theta_{\max}\approx41^\circ$. Each reduced dividing surface is a
two-sphere, and every lifted slice is $S^2\times S^1$ cut into two solid
tori, uniformly across the three families and the whole momentum interval.
The circulating type is nonetheless present in the model --- precisely at
the excluded center. At $p_\phi=0$ the planar FR1 \emph{rotates}: $\theta$
advances by $2\pi$ each period while $r$ oscillates twice --- the $2{:}1$
resonance --- and $p_\theta$ never vanishes; the planar OTS-PO is the
orbiting relative equilibrium at the centrifugal barrier, a rotation
likewise \cite{Mauguiere2014,Wiggins2026}. 
Only the planar TTS-PO librates, through the axis. For FR1 and OTS the type
classification is thus realized across the singular center rather than
across the families: circulating exactly on the excluded planar orbit,
librating on every member --- one more respect in which the generating orbit
differs in kind from the family it bounds. That difference --- membership
against boundary, fiber against separatrix --- has no analog in the
reduced theory, and it, rather than any distinction of surface topology, is
what organizes this model. We turn to it now.

\begin{figure}[t]
\centering
\includegraphics[width=\textwidth]{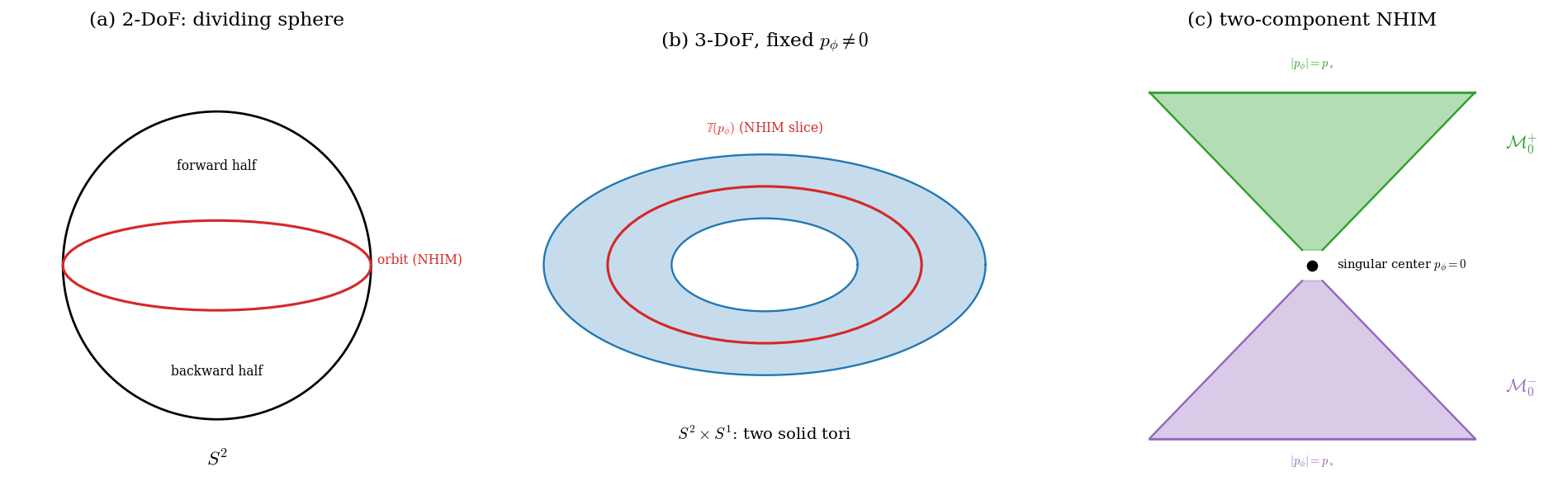}
\caption{Lifting the dividing surface from two to three degrees of freedom.
(a) In two degrees of freedom the dividing surface attached to a librating
periodic orbit is a two-sphere $S^2$, with the orbit (a NHIM) as a dividing
equator splitting it into forward and backward halves \cite{Mauguiere2016}.
(b) In three degrees of freedom, at fixed angular momentum $p_\phi\neq0$,
the azimuthal circle lifts the sphere to the product $S^2\times S^1$; the
NHIM slice is now the two-torus $\mathbb{T}(p_\phi)$ --- the equator, that
is, the reduced orbit, crossed with the azimuthal circle --- and it cuts
$S^2\times S^1$ into two solid tori $D^2\times S^1$, the forward and
backward halves. (c) Assembling over $p_\phi\in(-p_*,0)\cup(0,p_*)$ produces
a two-component NHIM, one component per precession sense
($\mathcal{M}_0^{\pm}$), each open at both ends: at the inner boundary, the
excluded singular center at $p_\phi=0$, where each transition state has its
own axis-crossing planar orbit --- the separatrix of
Section~\ref{sec:separatrix} --- and at the outer boundary $|p_\phi|=p_*$,
where normal hyperbolicity fails.}
\label{fig:lift}
\end{figure}

\subsection{The orbit FR1 as a separatrix, not a fiber}\label{sec:separatrix}
The families of Section~\ref{sec:families} are parametrized over the open
interval $0<|p_\phi|<p_*$; the orbit FR1 lies at $p_\phi=0$, outside both
components. Whether FR1 could nonetheless serve as a central fiber
completing $\mathcal{M}_0$ is decided by the limit $p_\phi\to0^+$, which we
now compute; the precise statements are
Propositions~\ref{prop:sep}--\ref{prop:blowup}, and the reading of the
marginal transverse multipliers in Section~\ref{sec:results} depends on the
answer. If the family closed up smoothly with FR1 as its $p_\phi=0$ fiber,
the polar turning states of the members --- whose positions converge to
FR1's axis-crossing point $(r_t,\theta=0)$, $r_t=3.6507$ --- would converge
to states of FR1. At a polar turning point both reduced momenta vanish
(hypothesis (H)(i)), so energy conservation there reads
$E=U(r_t,\theta_-;p_\phi)+V(r_t,\theta_-)$, and expanding the effective
centrifugal potential for small $\theta$ and inserting
$\theta_-=k\,p_\phi$ gives, as $p_\phi\to0$,
\begin{equation*}
E-V_{\mathrm{CH}}(r_t)\;=\;\frac{1}{2k^2}
\Big(\frac{1}{mr_t^2}+\frac{1}{I_x}\Big)\;=:\;\mathcal{E}^0 .
\end{equation*}
Read one way, this identity determines the constant of the linear law of
Section~\ref{sec:families}: $k=0.4619$, with $\mathcal{E}^0=1.1737$, and
$\rho_{\min}=r_tk\,|p_\phi|=1.686\,|p_\phi|$. Read the other way, it states
what the turning states converge to: their canonical momenta all vanish in
the limit --- $p_r=p_\theta=0$ exactly, $p_\phi\to0$ --- while their
kinetic energy tends to $\mathcal{E}^0$, because the coefficient of
$p_\phi^2$ in $U$ diverges as $\theta_-^{-2}$; the limiting states sit at
$(r_t,\theta=0)$ with zero momentum, reached along circulation about the
axis on the shrinking circle $\rho_{\min}$ at azimuthal frequency
$\dot\phi=2\mathcal{E}^0/|p_\phi|+O(1)$. On FR1 the axis crossing carries
$|p_\theta|=p_{\max}=2.165$. The turning states therefore do not converge
to states of FR1.

The dynamical name for this situation is a separatrix, and the contrast
promised in Section~\ref{sec:families} can now be drawn exactly. A
\emph{fiber}, as defined there, is a regular member of the parametrized
family: interchangeable with its neighbors, deformed smoothly into them by
varying $p_\phi$, with the manifold locally a product of the fiber and the
parameter interval. A \emph{separatrix} is not a member of a family but a
boundary between families of different character. The canonical picture is
the pendulum, whose phase portrait contains two families of qualitatively
different motion --- librations, oscillating between turning points, and
circulations, passing over the top --- divided by one exceptional orbit
whose own motion is unlike every neighbor on either side. Membership
against boundary is the whole distinction.

FR1 is a boundary in both available senses. In the parameter: $p_\phi=0$ is
the boundary between the two counter-precessing components
$\mathcal{M}_0^+$ and $\mathcal{M}_0^-$ --- the only value at which the
precession sense could change, and precisely the value excluded from both;
the conserved-quantity reading of the same fact, with $p_\phi=0$ as the
level set dividing the two senses, is drawn in Section~\ref{sec:persist}.
In the motion: every member of either family reflects at its polar turning
points --- for $p_\phi\neq0$ the effective centrifugal potential $U$ is an
infinite wall at the poles, and the orbit turns at $\theta_-(p_\phi)>0$
without reaching the axis --- whereas at $p_\phi=0$ the wall vanishes
identically and $\theta=0$ is simply the minimum of the bending potential
$V_0\sin^2\theta$, the collinear well direction. (This wall is unrelated to
the Chesnavich lock, which is the ridge of the same potential at the
equator.) The planar orbit crosses the axis with maximal $|p_\theta|$ ---
indeed it rotates, $\theta$ advancing monotonically
(Section~\ref{sec:dstopo}). Reflection on every fiber; crossing on the
boundary alone. The motion on FR1 differs in kind from the motion on every
member of the family it bounds --- the defining property of a separatrix,
exactly as the pendulum's separatrix passes over the top that every
libration turns back from.

The comparison with the pendulum also locates what is singular here. The
pendulum separatrix is the Hausdorff limit of its neighboring orbits, and
it is an object of the same kind as they are --- a solution curve. Here the
limit exists but is of the wrong kind. By Proposition~\ref{prop:sep}(b),
the Hausdorff limit of the family is the folded planar orbit together with
the two momentum segments
$\{(r_t,\theta{=}0\ \text{or}\ \pi,\,0,\,p_\theta,\,0):|p_\theta|\le
p_{\max}\}$: a singular set, not an orbit and not a torus. The members do
converge onto FR1 away from the axis; the failure is confined to a polar
layer of angular width $O(|p_\phi|)$, inside which the centrifugal force
diverges and smooth dependence on $p_\phi$ breaks down, and it is there
that the momenta sweep out the extra segments. What fails at the center is
therefore not attainment but fiber structure: the tori converge to no
torus, and the fibration has no fiber over $p_\phi=0$ that FR1 could be.
The sharp separation between orbit and family is directional and is made
explicit by the blow-up (Proposition~\ref{prop:blowup}): on the exceptional
divisor the family lands on two latitude circles, one per precession sense,
while the planar orbits lift to the equator, disjoint from both. FR1 is
excluded from $\mathcal{M}_0$ by definition --- it is not a member of
either family --- and no smooth completion exists that could admit it.

This is also the correct reading of the marginal transverse multipliers
$(+1,+1)$ of FR1 at $b=0$ (Section~\ref{sec:results}): at the equatorial
crossing the azimuthal-rotation and family directions span the transverse
plane $(Y,p_Y)$, and both are neutral; the identification degenerates at
the poles, where the rotation field vanishes (the apex), and the
marginality is the linear-stability signature of sitting at the singular
center rather than evidence that the out-of-plane direction is dynamically
inert. The orbit's genuine instability is the reduced hyperbolic pair
$\{24.57,\,0.0407\}$. The same statements hold for the entrance and the
reactive gate: each planar orbit crosses the axis through the identical
centrifugal structure --- the OTS-PO as a rotation at $r\approx13.4$, the
orbiting relative equilibrium, and the TTS-PO as a libration through the
axis over the mouth of the well --- so each family carries its own conical
center and its own separatrix orbit, with the per-family constants of
Proposition~\ref{prop:sep} ($k=0.651$ and $0.304$). For the TTS one
refinement applies: its planar orbit closes only after two circuits of the
limiting reduced arc, so the family's period and multiplier converge to
half and to the square root of the planar values
(Proposition~\ref{prop:sep}(b)).

\subsection{Persistence and the fate of the two components for $b>0$}\label{sec:persist}
Switching on the coupling breaks the azimuthal symmetry: for $b>0$ the
momentum $p_\phi$ is no longer conserved and $\mathcal{M}_0$ is no longer
invariant. What survives is governed by normal hyperbolicity, and the
hypothesis is verified where the manifold exists exactly, at $b=0$. Within
$\Sigma_E$ the manifold $\mathcal{M}_0$ is three-dimensional; at a point of
the torus $\mathbb{T}(p_\phi)$ its tangent space is spanned by the flow
direction, the azimuthal rotation $\partial_\phi$, and the family direction
$\partial_{p_\phi}$, so exactly two directions remain, and these --- the
normal bundle --- are the stable and unstable directions of the reduced
orbit, the pair $\{\lambda(p_\phi),\lambda(p_\phi)^{-1}\}$
(Proposition~\ref{prop:nhim}). The normal rate is $\ln\lambda(p_\phi)$ per
period, decreasing from $\ln 24.57\approx3.20$ at the inner edge to $0$ at
$p_*$; the tangential rate is zero, because the internal motion lies on
invariant tori and separates at most polynomially. Two exclusions define the
domain of the statement, one per boundary: the outer edge is excluded
because $\lambda\to1$ at $p_*$, where the normal rate vanishes; the center
is excluded because it belongs to neither component --- the components are
open there, and FR1 was never inside the manifold whose persistence is at
issue. On any closed sub-annulus $\varepsilon\le|p_\phi|\le p_*-\delta$ each
component is therefore compact and normally hyperbolic of every order, and
the standard theorems
\cite{Fenichel1971,HirschPughShub1977,Wiggins1994,Wiggins2016} give
persistence: for each smoothness class $C^r$ there is a range of $b$ on
which each component survives as a locally invariant $C^r$
manifold-with-boundary, together with its stable and unstable manifolds and
the dividing surfaces it carries (Proposition~\ref{prop:persist}).

The theorems do not name the range of $b$; the computed rates indicate it.
One identification must be kept straight, because ``plane'' is doing two
jobs. The normal directions of $\mathcal{M}_0$ lie \emph{within} each
reduced slice --- they are the hyperbolic directions of the reduced saddle
--- while the out-of-plane directions $(Y,p_Y)$ at the planar limit are
spanned by $\partial_\phi$ and $\partial_{p_\phi}$ and are \emph{tangent} to
$\mathcal{M}_0$ (Section~\ref{sec:separatrix}). The out-of-plane instability
that the coupling generates is therefore growth along the manifold, not
across it, and it is the quantity the tangential rate must be tested
against. The computed values: the multiplier of the planar orbit, which
monitors the normal rate at the manifold's inner edge, stays at
$\ln\lambda\approx3.2$ per period essentially unchanged over the range
studied, while the out-of-plane growth does not exceed $\ln\approx1.2$ per
period (Section~\ref{sec:results}); the gap $3.2>2\times1.2$ required for
class $C^2$ stays open through $b\approx0.3$. For the other two families:
$\ln\lambda_{\mathrm{TTS}}=5.4$ per period against a tangential rate of at
most $0.25$ for $b\le0.4$, the widest gap of the three; and
$\ln\lambda_{\mathrm{OTS}}=0.12$ against a perturbation that is itself
negligible at the orbit, the coupling at $r\approx13.4$ being proportional
to $V_0(13.4)\sim10^{-66}$.

The orbit FR1 requires no persistence theorem: it lies in the reaction
plane for every $b$ (Section~\ref{sec:model}(iii)), so we continue it from
$b=0$ and compute its stability directly --- the central computation of
Section~\ref{sec:results}: transversely hyperbolic for $b>0$ outside the
elliptic window, with the period-doubling at $b_c\approx0.63$.

The remaining question is what the symmetry breaking does to the structure
the manifolds organize, and here the word \emph{separatrix} carries its
classical meaning. At $b=0$ the sign of $p_\phi$ is conserved, and the level
set $p_\phi=0$ separates two qualitatively distinct families of motion ---
the two precession senses --- exactly as the pendulum separatrix separates
its two senses of rotation; the planar orbit lies in the dividing level set,
and its own motion differs in kind from every member of the families it
bounds (rotation against libration, Section~\ref{sec:dstopo}). For $b>0$
nothing conserves $p_\phi$:
\[
\dot p_\phi=-\frac{\partial H}{\partial\phi}=3bV_0(r)\sin^3\theta\sin3\phi,
\]
identically zero at $b=0$ and on the plane, nonzero as soon as a trajectory
acquires out-of-plane displacement. The level set no longer separates, and
the two senses communicate. The computation shows the opening directly: a
trajectory launched on a torus $\mathbb{T}(p_\phi)$ of the family
holds $p_\phi$ rigid to the printed precision at $b=0$, has $p_\phi$ spread
over $0.14$ at $b=0.05$, and at $b=0.30$ spreads over more than unity ---
crossing $p_\phi=0$ within twelve roaming periods, with energy conserved to
one part in $10^9$. No invariance is contradicted: the persisted sub-annuli
are locally invariant, manifolds-with-boundary through whose edges
trajectories may leave, and the crossing is exit through the inner edge and
passage across the region where the excluded center was. The invariant tori
that foliate each component at $b=0$ do not survive individually; the drift
in $p_\phi$ that replaces them is the out-of-plane escape of
Section~\ref{sec:results} seen from the manifold's side, and it is what the
trajectory classification of Section~\ref{sec:transport} measures.

Two remarks close the construction. Roaming has no index-one saddle of the
potential: the manifolds sit over the roaming shelf, the centrifugal
barrier, and the mouth of the well, not over critical points, so the
normal-form construction of a NHIM at a saddle
\cite{WaalkensSchubertWiggins2008} does not apply; the manifolds are instead
exhibited in the symmetric integrable limit and continued, which is why the
$b=0$ verification is the load-bearing step. And the period-doubling at
$b_c$ is a bifurcation of the flow \emph{on} the manifold, not of the
manifold: it reorganizes the internal dynamics while the normal directions
keep their gap. The rigorous statements --- existence and smoothness of the
sub-annuli, their normal hyperbolicity of every order, the leaf topology,
the separatrix structure at the center, and persistence --- are collected
and proved in Section~\ref{sec:proofs}. The classification of trajectories
at finite $b$ is the subject of Section~\ref{sec:transport}; the flux
question is not treated in this paper; an entropy-based flux analysis for two-degree-of-freedom Chesnavich-type models is given in \cite{Wiggins2026}.

\subsection{Rigorous statements and proofs}\label{sec:proofs}
The construction of the preceding subsections rests on one input that we
establish numerically and otherwise reason from rigorously. We isolate it as
a standing hypothesis, state the structural results as propositions, and
prove them. Throughout, $\Sigma_E\subset\mathbb{R}^6$ is the energy surface
$\{H_0=E\}$ at $b=0$ (five-dimensional), $U(r,\theta;p_\phi)$ is the
effective centrifugal potential of \eqref{eq:Hred}, and $\gamma_{p_\phi}$
denotes the symmetric reduced periodic orbit of \eqref{eq:Hred} at angular
momentum $p_\phi$ in the FR1 (respectively OTS or TTS) family. The planar
orbits are not members of the families: at $p_\phi=0$ the planar FR1 and
OTS-PO rotate and have no turning points (Section~\ref{sec:dstopo}), so each
family lives on an open interval and the planar orbit enters only as its
singular limit (Proposition~\ref{prop:sep}).

\begin{hypothesis}
At $E=0.5$ and $b=0$, for each family $F\in\{\mathrm{FR1},\mathrm{OTS},
\mathrm{TTS}\}$ there is $p_*^F>0$ such that on the open interval
$(0,p_*^F)$ the reduced Hamiltonian \eqref{eq:Hred} admits a periodic orbit
$\gamma_{p_\phi}$ depending $C^\infty$ on $p_\phi$, with the following
properties.
\begin{enumerate}
\item[\textnormal{(i)}] $\gamma_{p_\phi}$ is a brake orbit: it has exactly
two turning points, at each of which both reduced momenta vanish,
$p_r=p_\theta=0$.
\item[\textnormal{(ii)}] The reduced monodromy matrix of $\gamma_{p_\phi}$
has exactly one pair of multipliers off the unit circle,
$\{\lambda(p_\phi),\lambda(p_\phi)^{-1}\}$, with $\lambda$ real and greater
than one at every $p_\phi\in(0,p_*^F)$.
\item[\textnormal{(iii)}] Along the family, $dT/dE\neq0$ and
$d\Delta\phi/dp_\phi\neq0$, where $T(p_\phi)$ is the period and
$\Delta\phi(p_\phi)$ the azimuthal advance of the reconstructed orbit over
one period.
\item[\textnormal{(iv)}] At the endpoints: as $p_\phi\to0^+$, $\lambda$ and
$T$ converge to the planar values for FR1 and OTS, and to
$\sqrt{\lambda_{\mathrm{planar}}}$ and $T_{\mathrm{planar}}/2$ for TTS. As
$p_\phi\to p_*^{F-}$: for FR1, $\lambda\to1$ at finite orbit amplitude; for
OTS and TTS, $\lambda$ remains bounded away from $1$ and the orbit's
amplitude tends to zero, the family terminating on a saddle-center
equilibrium of the reduced system.
\end{enumerate}
\end{hypothesis}

\noindent Local existence and smoothness are not assumptions but
consequences of the implicit function theorem: at any $p_\phi$ with
$\lambda\neq1$ the symmetric-return (Poincar\'e) map has a transverse,
non-degenerate fixed point, which persists and varies smoothly in $p_\phi$.
What the hypothesis adds beyond the IFT is global and quantitative, and it
is exactly what the computations establish
(Appendix~\ref{app:verification}). For FR1: $\lambda$ decreases from
$24.57$ to $1$, with $p_*=1.799$; the unit eigenvalue of the monodromy
matrix has multiplicity four with geometric multiplicity two, verified by
spectral projector, and (iii) is verified directly ($\Delta\phi$ strictly
decreasing from $2\pi$ at the center, $dT/dE=-3.16$ at $p_\phi=0.5$). For
OTS: $\lambda$ increases from $1.131$ to $1.269$, and the family terminates
at $p_*=2.150$ on the azimuthal relative equilibrium at the centrifugal
barrier ($r=9.600$); the endpoint multiplier is predicted by the
equilibrium's linearized rates, $e^{\mu_r\,2\pi/\omega_\theta}=1.269$, and
matched by the family. For TTS: $\lambda$ increases from its inner-edge limit
$\sqrt{220.46}=14.848$ to $15.90$, and the family terminates at
$p_*=1.317$ on the saddle-center equilibrium that the centrifugal term
creates in the mouth of the well ($r=2.446$, $\theta=29.4^\circ$; the bare
potential has no critical point there), again with the endpoint multiplier
$e^{\mu_r\,2\pi/\omega_\theta}=15.90$ matched by the family. We therefore
treat (H) as established for the model and prove the rest from it.

\begin{proposition}[The two-component manifold]\label{prop:annulus}
Under \textnormal{(H)}, fix $0<\varepsilon<p_*$ and $\delta>0$ with
$\varepsilon\le p_*-\delta$. The set
\[
\mathcal{M}_0^{[\varepsilon,\delta]}
=\bigcup_{\varepsilon\le|p_\phi|\le p_*-\delta}\ \bigcup_{\phi\in S^1}\
\big\{\gamma_{p_\phi}\ \text{rotated by }\phi\big\}
\]
is the disjoint union of two compact, $C^\infty$, flow-invariant
three-dimensional submanifolds-with-boundary of $\Sigma_E$ --- one for each
sign of $p_\phi$ --- each with two boundary tori at $|p_\phi|=\varepsilon$
and $|p_\phi|=p_*-\delta$, hence four boundary tori in total.
\end{proposition}
\begin{proof}
For $|p_\phi|\ge\varepsilon>0$ the effective centrifugal potential
$U(r,\theta;p_\phi)$ in \eqref{eq:Hred} diverges as $\theta\to0,\pi$, so
$\gamma_{p_\phi}$ has turning angles bounded away from the poles and lies in
the open off-axis region $\{0<\theta<\pi\}$, where the azimuthal $SO(2)$
action $\phi\mapsto\phi+\alpha$ is free. Parametrise
$\mathcal{M}_0^{[\varepsilon,\delta]}$ by $(s,\phi,p_\phi)$,
$s\in\mathbb{R}/T(p_\phi)\mathbb{Z}$ the phase along $\gamma_{p_\phi}$. By
(H) the map $(s,\phi,p_\phi)\mapsto$ phase point is $C^\infty$; it is an
immersion ($\partial_s$ is the reduced flow, $\partial_\phi$ the rotation,
$\partial_{p_\phi}$ the transverse family direction, and $p_\phi$ is a
submersion onto the momentum interval) and injective (distinct $|p_\phi|$
lie on distinct values of the conserved momentum; at fixed $p_\phi$ distinct
$(s,\phi)$ give distinct points since the action is free and
$\gamma_{p_\phi}$ is a simple closed curve), hence an embedding of the
compact manifold-with-boundary $[\varepsilon,p_*-\delta]\times S^1\times
S^1$ for $p_\phi>0$, and independently of its $-p_\phi$ copy for
$p_\phi<0$; the two embeddings have disjoint images since $p_\phi$
separates them. Invariance: $\gamma_{p_\phi}$ is invariant under the
reduced flow, the rotation is a symmetry of $H_0$, and $p_\phi$ is
conserved, so the full flow carries each lifted torus to itself; the union
is invariant. The boundary is the union of four tori, two per
sign-component.
\end{proof}

\begin{proposition}[Floquet spectrum and normal hyperbolicity]
\label{prop:nhim}
Under \textnormal{(H)}, for each $p_\phi\in(0,p_*)$ the
relative-periodic-orbit monodromy matrix
$M_{\mathrm{RPO}}=R(-\Delta\phi)\,D\Phi_T$ has spectrum
$\{\lambda,\lambda^{-1},1,1,1,1\}$, the eigenvalue $1$ having algebraic
multiplicity four and geometric multiplicity two (two $2\times2$ Jordan
blocks: flow--energy and symmetry--momentum). Within $\Sigma_E$ the
manifold $\mathcal{M}_0^{[\varepsilon,\delta]}$ has tangent bundle spanned
by the flow, the azimuthal rotation, and the family direction, all with
zero Lyapunov exponent, and normal bundle the pair
$\{\lambda,\lambda^{-1}\}$ --- the normal pair of the reduced orbit ---
with rate $\nu(p_\phi)=\ln\lambda(p_\phi)/T(p_\phi)$ bounded away from
zero. Consequently $\mathcal{M}_0^{[\varepsilon,\delta]}$ is $r$-normally
hyperbolic for every $r\ge1$.
\end{proposition}
\begin{proof}
Four eigenvalue-$1$ directions are forced by structure. (i) The flow
direction $X_{H_0}$ is a fixed vector of the time-$T$ map, hence a $+1$
eigenvector. (ii) The period varies with energy, $dT/dE\neq0$ by
\textnormal{(H)(iii)}, so the energy-gradient direction is a generalised
$+1$ eigenvector partnering the flow in a $2\times2$ Jordan block (the
standard trivial pair). (iii) The infinitesimal azimuthal rotation
$\xi_\phi=\partial_\phi$ generates the $SO(2)$ symmetry, which commutes
with the flow and maps $\gamma_{p_\phi}$ to its own rotation; hence
$\xi_\phi$ is a $+1$ eigenvector. (iv) The family direction
$\partial_{p_\phi}$ is mapped by the linearized return to itself plus a
multiple of $\xi_\phi$, because $d\Delta\phi/dp_\phi\neq0$ by
\textnormal{(H)(iii)}; this is a generalised $+1$ eigenvector partnering
$\xi_\phi$ in a second $2\times2$ Jordan block (the symmetry--momentum
pair). These four directions span the generalised eigenspace at $1$. By
\textnormal{(H)(ii)} the reduced normal pair contributes the only
multipliers off the unit circle, so the algebraic multiplicity of $1$ is
exactly four and the geometric multiplicity two. On $\Sigma_E$ the energy
direction is quotiented, leaving the flow, rotation, and family as the
three tangent directions of $\mathcal{M}_0$ --- all neutral, tangential
exponent $0$ --- and the normal pair as the two-dimensional normal bundle,
with $\nu(p_\phi)\in[\nu_{\min},\nu_{\max}]\subset(0,\infty)$ on the closed
sub-annulus. Normal hyperbolicity of order $r$ requires the normal rate to
exceed $r$ times the tangential rate; the two Jordan blocks contribute only
polynomial (sub-exponential) drift along $\mathcal{M}_0$, so the tangential
Lyapunov rate is exactly $0$ and the inequality holds for every $r\ge1$.
\end{proof}

\begin{proposition}[Dividing-surface topology]\label{prop:topology}
Under \textnormal{(H)}, for each $p_\phi\in(0,p_*)$ the dividing surface
anchored on the $p_\phi$-slice of $\mathcal{M}_0$ is homeomorphic to
$S^2\times S^1$, and the slice of the NHIM (a two-torus) separates it into
two components, each homeomorphic to the solid torus $D^2\times S^1$. The
same holds for the OTS and TTS families. The circulating (toroidal)
dividing surface of \textnormal{\cite{Mauguiere2016}} does not occur for
any member of the three families.
\end{proposition}
\begin{proof}
Written with $(q_1,q_2)=(r,\theta)$, the reduced Hamiltonian
\eqref{eq:Hred} is of the two-degree-of-freedom form treated in
\cite{Mauguiere2016}, the centrifugal terms folding into the effective
potential $U(r,\theta;p_\phi)$. By \textnormal{(H)(i)} the reduced orbit
$\gamma_{p_\phi}$ is a brake orbit: at its two turning points both reduced
momenta vanish, which is precisely the degeneration of the momentum ellipse
that defines the librating (type-1) case; by \textnormal{(H)(ii)} it is
hyperbolic. The construction of \cite{Mauguiere2016} then gives a reduced
dividing surface $D_{\mathrm{red}}\cong S^2$ on which $\gamma_{p_\phi}$ is
a separating equator through the two degenerate points. Restoring the
azimuth: at fixed $(E,p_\phi)$ the full slice fibers over $D_{\mathrm{red}}$
with fiber the azimuthal circle $S^1$. Because $\gamma_{p_\phi}$ and
$D_{\mathrm{red}}$ lie in the off-axis region, the angle $\phi$ is a global
coordinate on the fiber, so the bundle is trivial and the slice is
$D_{\mathrm{red}}\times S^1\cong S^2\times S^1$, with NHIM slice
$(\text{equator})\times S^1\cong T^2$. Since $S^2$ minus an equator is two
open disks,
$(S^2\times S^1)\setminus T^2=(D^2\times S^1)\sqcup(D^2\times S^1)$. For
OTS and TTS the generating orbit is likewise a brake orbit by
\textnormal{(H)(i)} --- the OTS-type members librating across the equator,
the TTS-type members on one side of it --- so $D_{\mathrm{red}}\cong S^2$
and the conclusion is identical. The toroidal case of \cite{Mauguiere2016}
requires a circulating generating orbit; for $p_\phi\neq0$ no coordinate of
the reduced system is periodic (the centrifugal barrier walls off both
poles and $r$ is not an angle), so no member of any family circulates and
that case does not arise --- the circulating type is realized in this model
only by the planar FR1 and OTS-PO at the excluded center
(Section~\ref{sec:dstopo}). The full four-dimensional dividing surface in
$\Sigma_E$ is the union of these three-dimensional leaves over
$p_\phi\in(0,p_*)$ together with the corresponding $-p_\phi$ union; the fibration over
$p_\phi$ and the leaf-by-leaf description are exact only at $b=0$, where
$p_\phi$ is conserved. For sufficiently small $b>0$, the dividing surfaces carried by each
compact interior subannulus deform smoothly with the corresponding locally
persisted NHIM branch (Proposition~\ref{prop:persist}); no global
continuation through the singular $p_\phi=0$ limit is asserted, and $p_\phi$
survives only as a local label on each branch, not as an invariant that
foliates the surface.
\end{proof}

\begin{proposition}[Separatrix structure at the center]\label{prop:sep}
Under \textnormal{(H)}, let $r_t$ denote the radius of polar approach of
the family and $\mathcal{E}^0=E-V_{\mathrm{CH}}(r_t)>0$ the kinetic energy
there. As $p_\phi\to0^+$:
\begin{enumerate}
\item[\textnormal{(a)}] the polar turning angle satisfies
$\theta_-(p_\phi)=k\,p_\phi+O(p_\phi^2)$ with
\[
k=\Big[\tfrac{\,I_x^{-1}+m^{-1}r_t^{-2}\,}{2\,\mathcal{E}^0}\Big]^{1/2};
\]
hence $\rho_{\min}(p_\phi)=r_t\sin\theta_-=r_tk\,p_\phi+O(p_\phi^2)$ is
linear in $p_\phi$, and the family meets the symmetry axis in a double cone
with a genuine conical (non-manifold) singularity at the apex;
\item[\textnormal{(b)}] the turning states converge to the phase points
$(r_t,\theta=0\ \text{or}\ \pi,\,p=0)$, which do not lie on the planar
orbit $\gamma_0$ (which crosses the axis with $|p_\theta|=p_{\max}=
\sqrt{\mathcal{E}^0/\kappa(r_t)}$, $\kappa=\tfrac1{2I_x}+\tfrac1{2mr^2}$),
while away from the axis the members converge to the folded planar curve.
The Hausdorff limit of the family is therefore the folded planar orbit
together with the two momentum segments
$\{(r_t,\theta{=}0\ \text{or}\ \pi,\,0,\,p_\theta,\,0):|p_\theta|\le
p_{\max}\}$ --- a singular set strictly containing $\gamma_0$. The family
consequently has no $p_\phi=0$ fiber: the torus fibration degenerates at
the center, and $\gamma_0$ is the separatrix carrier in the sense of
Section~\ref{sec:persist} --- the zero level set of the conserved momentum,
dividing the two precession senses --- its own motion differing in kind
from every member. For FR1 and OTS the members' periods and multipliers
converge to the planar values; for TTS the planar orbit closes only after
two circuits of the limiting reduced arc, so
$T(p_\phi)\to T_{\mathrm{planar}}/2$ and
$\lambda(p_\phi)\to\sqrt{\lambda_{\mathrm{planar}}}$.
\end{enumerate}
\end{proposition}
\begin{proof}
At a polar turning point both reduced momenta vanish
(\textnormal{(H)(i)}), so the energy relation \eqref{eq:Hred} reads
$E=U(r_t,\theta_-;p_\phi)+V(r_t,\theta_-)$. As $\theta\to0$,
$U(r,\theta;p_\phi)=\tfrac{p_\phi^2}{2\theta^2}\big(I_x^{-1}+m^{-1}r^{-2}
\big)+O(p_\phi^2)$ and $V(r_t,\theta)=V_{\mathrm{CH}}(r_t)+O(\theta^2)$.
Substituting and keeping leading order,
\[
\frac{p_\phi^2}{2\theta_-^2}\big(I_x^{-1}+m^{-1}r_t^{-2}\big)
=\mathcal{E}^0+O(\theta_-^2),
\]
which gives $\theta_-=k\,p_\phi+O(p_\phi^2)$ with $k$ as stated. The set
traced near the apex, $\{X^2+Y^2=(r_tk)^2p_\phi^2\}$, is a double cone,
which is not a smooth manifold at the apex (its tangent cone is the cone
itself, not a plane). For (b): the turning states have both reduced
momenta zero and $p_\phi\to0$, so they converge to
$(r_t,\theta{=}0,\,0,0,0)$; on $\gamma_0$ the axis crossing carries
$|p_\theta|=p_{\max}\neq0$, so this limit point is not on $\gamma_0$.
Inside the polar layer the momenta sweep the full segment: at the point of
$\gamma_{p_\phi}$ with $\theta_c=\theta_-/\sqrt{1-c^2}$, $c\in(0,1)$, the
energy relation with the same expansion gives
$p_\theta^2\,\kappa(r_t)=\mathcal{E}^0\big(1-\theta_-^2/\theta_c^2\big)
+o(1)=\mathcal{E}^0c^2+o(1)$, so $p_\theta\to c\,p_{\max}$, while
$r\to r_t$ and $p_r\to0$ (the radial variation across the layer is
$O(\theta^2)$); every point of the segment is thus a limit of family
points. Away from any fixed neighborhood of the poles the reduced vector
field depends smoothly on $p_\phi$ down to $0$, so the members converge
there to the folded planar curve. Together these give the stated Hausdorff
limit, which contains $\gamma_0$ and strictly more; since the limit is not
a torus --- not even a manifold --- the fibration has no fiber over
$p_\phi=0$. The covering statement is forced by the pole count: the planar
FR1 and OTS-PO cross both poles per period, so one circuit of the folded
curve closes them; the planar TTS-PO crosses one pole per half-swing,
exchanging meridional half-planes at each crossing, so two circuits are
required, halving the limiting period and taking the square root of the
multiplier.
\end{proof}

\noindent The constants are quantitatively confirmed for all three
families (Appendix~\ref{app:verification}): measured $\theta_-/p_\phi$
against the formula for $k$ gives $0.4619$ (FR1: $r_t=3.6507$,
$\mathcal{E}^0=1.1736$), $0.651$ (OTS: $r_t=13.37$,
$\mathcal{E}^0=0.5035$), and $0.304$ (TTS: $r_t=2.6444$,
$\mathcal{E}^0=3.105$); for FR1, $\rho_{\min}/|p_\phi|=1.686$ against
$r_tk=1.6863$; and the TTS covering relation fixes the inner-edge limits exactly,
$T(0^+)=T_{\mathrm{planar}}/2=1.1139$ and $\lambda(0^+)=\sqrt{220.46}=14.848$;
the smallest-momentum member computed ($p_\phi=0.05$) has $T=1.1140$ and
$\lambda=14.84$, approaching both limits to three digits.

\begin{proposition}[The blow-up separates the directions of approach]
\label{prop:blowup}
Let $\beta:\widetilde{N}\to N$ be the oriented (spherical) blow-up of the
apex in the three transverse directions $(X,Y,p_\phi)$. The proper
transform of the cone of Proposition~\ref{prop:sep} is two disjoint smooth
circles on the exceptional divisor $S^2$ --- the latitude circles
$\hat p=\pm c/\sqrt{1+c^2}$, $c=r_tk$, one per precession sense --- while
every planar orbit lies in $\{p_\phi=0\}$ and lifts to the equator
$\hat p=0$: the orbit $\gamma_0$ meets the divisor in two antipodal
equatorial points (its approach and departure directions), and its rotated
copies fill the equator. The family and the planar orbits therefore reach
the apex along separated directions.
\end{proposition}
\begin{proof}
In blow-up coordinates $(X,Y,p_\phi)=\rho\,(\hat X,\hat Y,\hat p)$ with
$(\hat X,\hat Y,\hat p)\in S^2$ and $\rho\ge0$, the cone
$X^2+Y^2=c^2p_\phi^2$ becomes $\hat X^2+\hat Y^2=c^2\hat p^2$, i.e.\
$\hat p=\pm c/\sqrt{1+c^2}$: two latitude circles, symmetric about the
equator, disjoint from it and from each other for $c>0$, and smooth. They
are the proper transforms of the two nappes (the $\pm p_\phi$ components). The
orbit $\gamma_0$ approaches the apex within $\{p_\phi=0\}$, crossing the
axis along a meridional direction with $(X,Y)\to0$; its lift has
$\hat p=0$ and meets the equator at the two antipodal points given by the
azimuths of approach and departure, and the $SO(2)$-rotated copies of
$\gamma_0$ sweep the whole equator. Since the latitude circles sit at
$\hat p\neq0$, the two sets are disjoint.
\end{proof}

\begin{proposition}[Persistence]\label{prop:persist}
For each $r\ge1$ there is $b_0(r)>0$ such that for $|b|<b_0(r)$ each of the
two components of each closed sub-annulus
$\mathcal{M}_0^{[\varepsilon,\delta]}$ persists as a $C^r$ normally
hyperbolic \emph{locally} invariant manifold-with-boundary
$\mathcal{M}_b^{[\varepsilon,\delta]}$, together with its $C^r$ local
stable and unstable manifolds, on which the four-dimensional dividing
surfaces of Proposition~\ref{prop:topology} deform smoothly. The orbit FR1
persists independently, as an exact periodic orbit in the reaction plane,
and is not contained in $\mathcal{M}_b^{[\varepsilon,\delta]}$.
\end{proposition}
\begin{proof}
By Proposition~\ref{prop:nhim}, $\mathcal{M}_0^{[\varepsilon,\delta]}$ is
compact and $r$-normally hyperbolic for every $r$, the tangential rate
being $0$. At $b=0$ the boundary tori are invariant ($p_\phi$ is
conserved), so the manifold is neither overflowing nor inflowing; this is
handled in the standard way by modifying the vector field in a collar of
the boundary so that the manifold becomes overflowing invariant
(respectively inflowing, for the stable version), the modification being
supported away from the interior. The persistence theorem for overflowing
invariant manifolds \cite{Fenichel1971,HirschPughShub1977,Wiggins1994} then
yields, for perturbations that are $C^r$-small, a nearby $C^r$ manifold
with the same dimension and normal splitting, carrying its local stable and
unstable manifolds; undoing the collar modification, the persisted object
is locally invariant for the original flow --- trajectories may leave only
through its boundary. The perturbation $H_b-H_0=bV_0(r)\sin^3\theta\cos3
\phi$ is $C^\infty$ and $O(b)$ in every $C^r$ norm on the compact
sub-annulus, so the smallness requirement fixes $b_0(r)>0$; the required
smallness grows with $r$, whence the dependence of $b_0$ on $r$. The plane
$Y=0$ is invariant for all $b$ (Section~\ref{sec:model}), so FR1 persists
there as an exact periodic orbit. Finally,
$\mathcal{M}_b^{[\varepsilon,\delta]}$ lies in an $O(b)$ neighborhood of
$\mathcal{M}_0^{[\varepsilon,\delta]}$, which is bounded away from the
plane $p_\phi=0$ containing FR1 by the distance corresponding to
$|p_\phi|\ge\varepsilon$; for $|b|$ small the neighborhood does not reach
that plane, so FR1 is not contained in either component.
\end{proof}

\section{Trajectory classification and the out-of-plane channel}\label{sec:transport}
This section carries out the transport computation that Section~\ref{sec:po} sets up: sampling the entrance dividing surface, classifying trajectories against the three surfaces, and recording class fractions and gap-time statistics. Two protocols are fixed at the outset. The coupling takes the two values $b=0$ and $b=0.3$, and the organizing contrast is between them: at $b=0$ the momentum $p_\phi$ is conserved and the dynamics decomposes into a one-parameter family of two-degree-of-freedom subsystems; at $b=0.3$ it does not. The energy is $E=0.5\ \mathrm{kcal\,mol^{-1}}$, as everywhere in this paper; the class fractions alone are in addition recomputed at five higher energies, up to $E=2.0\ \mathrm{kcal\,mol^{-1}}$, for one purpose --- to find where the effect of the coupling dies away. Reactive fluxes are not computed in this paper; an entropy-based flux analysis for two-degree-of-freedom Chesnavich-type models is given in \cite{Wiggins2026}.

\subsection{The classification protocol}\label{sec:protocol}
Initial conditions are sampled on the inward-crossing half ($p_r<0$) of the entrance section, the sphere $r_{\mathrm{OTS}}=13.386\ \AAA$ --- the radial realization of the OTS dividing surface. This radius is that of the planar orbiting transition state, and it is the radius at which Maugui\`ere, Collins, Ezra, Farantos, and Wiggins place the entrance dividing surface of the planar problem \cite{Mauguiere2014}; the value carries over, while the surface at this radius is here four-dimensional rather than two-dimensional. Each trajectory is integrated until one of two things happens. Either $r$ falls below $r_{\mathrm{react}}=2.0\ \AAA$ --- inside the ridge crest $r_c=2.2\ \AAA$ and below the radial range of the TTS orbit (Section~\ref{sec:dstopo}) --- and the trajectory is reactive: it is captured into the well. Or the trajectory returns outward through the entrance surface and dissociates: it is non-reactive. Along the way, its crossings of the classifier (FR1) surface are counted. On a fixed-$p_\phi$ leaf the reduced FR1 orbit librates, its dividing surface is a two-sphere whose equator is the orbit (Section~\ref{sec:dstopo}), and a crossing of that surface is a passage through the radial bottleneck of the roaming shelf. The count is therefore taken as the number of crossings of the radius $r_{\mathrm{FR1}}=3.40\ \AAA$; Section~\ref{sec:robust} shows that the principal redistribution is robust over the tested range of placements of this radius.

Five classes result. A reactive trajectory ends inside the classifier, so it crosses an odd number of times: once (direct reactive) or three or more times (roaming reactive). A non-reactive trajectory ends outside, so it crosses an even number of times: twice (direct non-reactive), four or more times (roaming non-reactive), or not at all, turning around before ever reaching the classifier. Maugui\`ere, Collins, Ezra, Farantos, and Wiggins note this last case as possible but do not observe it in the planar model \cite{Mauguiere2014}. It does occur here: no such trajectory appears at $E=0.5\ \mathrm{kcal\,mol^{-1}}$, but they reach $2$--$3\%$ of the ensemble at $E=2.0\ \mathrm{kcal\,mol^{-1}}$, and we count them as their own class rather than folding them into the roaming count. No trajectory in any ensemble remained unclassified at the integration limit.

The sample is uniform in the four canonical surface coordinates $(\theta,\phi,p_\theta,p_\phi)$ over the energetically allowed region, with $p_r<0$ fixed by $H=E$. Appendix~\ref{app:classify} verifies that this uniform sample weights each initial condition by the rate at which trajectories cross the surface there --- the correct ensemble for an incoming stream. Because $p_\phi$ is one of the sampled coordinates, every value of $p_\phi$ enters the ensemble in its natural proportion, and no separate average over $p_\phi$ is required. At $b=0$ the sample decomposes over the conserved-$p_\phi$ leaves; at $b=0.3$ it does not.

\subsection{The classification plane and the $p_\phi$ stack}\label{sec:planes}
Figure~\ref{fig:planes} shows the trajectory class over the $(\theta,p_\theta)$ plane of entrance initial conditions, at $E=0.5\ \mathrm{kcal\,mol^{-1}}$, for four settings of $(b,p_\phi)$. Panel (a), the in-plane leaf $b=0,\ p_\phi=0$, reproduces the planar classification: diagonal bands of direct-reactive trajectories cut by strips of roaming and of direct non-reactive trajectories, $88.1\%$ direct and $11.9\%$ roaming on the computed grid. Panel (b), still $b=0$ but on the leaf $p_\phi=0.8$, shows that the leaves are far from interchangeable. The roaming fraction roughly doubles, to $24.3\%$; the band geometry changes from diagonal stripes to concentric shells; and the polar region becomes energetically forbidden, excluded by the centrifugal part of the effective potential $U(r,\theta;p_\phi)$ in Eq.~\eqref{eq:Hred}. The cylindrically symmetric three-degree-of-freedom problem is therefore a one-parameter family of two-degree-of-freedom subsystems that differ substantially from one another, and the microcanonical fractions are the $p_\phi$ average over this family, not any single member.

Panels (c) and (d) isolate the effect of the coupling. Panel (c), $b=0.3$ on the in-plane leaf $p_\phi=0$, is nearly identical to panel (a): a launch with $p_\phi=0$ and $\phi=0$ lies in the reaction plane $\{Y=0,\,p_Y=0\}$, which is invariant for every $b$ (Section~\ref{sec:po}), so such trajectories never sample the out-of-plane direction. Sampling restricted to the reaction plane, as in the planar studies, would therefore miss the three-degree-of-freedom effect entirely. Panel (d), $b=0.3$ on the off-plane leaf $p_\phi=0.8$, is where the coupling acts. Relative to panel (b), the roaming non-reactive class grows by $53\%$, the roaming reactive and direct non-reactive classes shrink (by $25\%$ and $10\%$), and the direct reactive class rises slightly ($+5\%$): the coupling redistributes the population toward roaming escape without closing the direct reactive channel.

\begin{figure}[tbp]
\centering
\includegraphics[width=\textwidth]{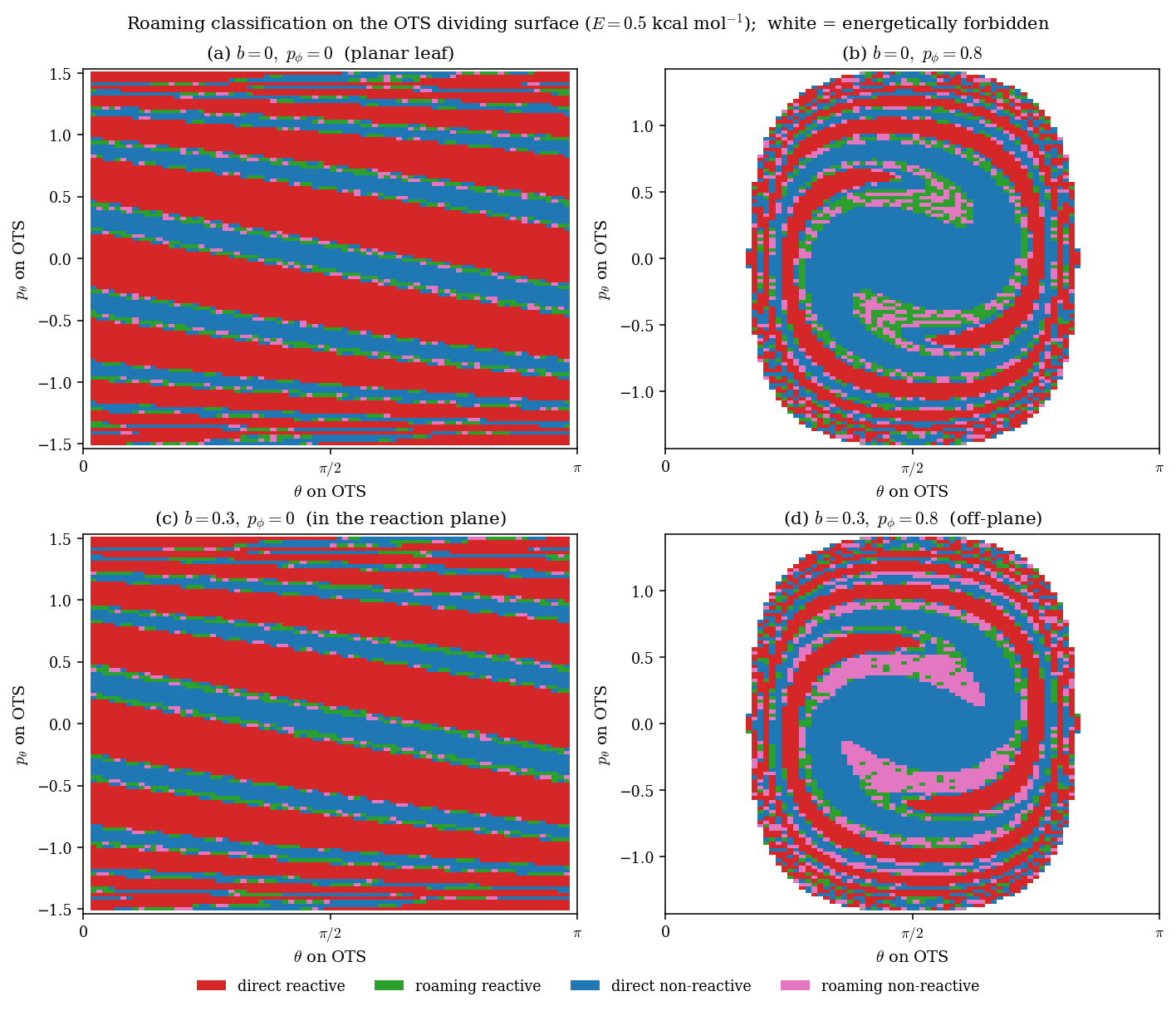}
\caption{Trajectory classification on the entrance dividing surface over the $(\theta,p_\theta)$ plane of incoming initial conditions ($\phi=0$), $E=0.5\ \mathrm{kcal\,mol^{-1}}$; white is energetically forbidden. (a) $b=0,\ p_\phi=0$: the planar leaf, $88.1\%$ direct, $11.9\%$ roaming. (b) $b=0,\ p_\phi=0.8$: roaming $24.3\%$, the polar regions excluded by the centrifugal part of $U$ in Eq.~\eqref{eq:Hred}. (c) $b=0.3,\ p_\phi=0$: in-plane, hence insensitive to the coupling --- nearly identical to (a) because the launch lies in the reaction plane $\{Y=0,\,p_Y=0\}$. (d) $b=0.3,\ p_\phi=0.8$: the roaming non-reactive class (pink) grows by $53\%$ at the expense of the roaming reactive and direct non-reactive classes. Classes: direct reactive (red), roaming reactive (green), direct non-reactive (blue), roaming non-reactive (pink).}
\label{fig:planes}
\end{figure}

\subsection{Microcanonical fractions versus energy}\label{sec:fractions}
Figure~\ref{fig:fractions} shows the microcanonical fraction of each class against energy, for $b=0$ and $b=0.3$, with sample sizes of $15000$ per coupling value at $E=0.5$, $5000$ at $E=0.7$, and $1200$ at each higher energy. Three findings. First, the $p_\phi$-averaged ensemble is far less reactive than the planar slice: at $E=0.5\ \mathrm{kcal\,mol^{-1}}$ the direct-reactive fraction of the incoming ensemble is $0.36$, against $0.61$ on the planar leaf, and at $b=0$ the direct non-reactive fraction rises from $0.41$ at $E=0.5$ to $0.66$ at $E=2.0\ \mathrm{kcal\,mol^{-1}}$ while the roaming reactive fraction collapses from $0.092$ to $0.001$. The planar problem is not representative of the three-degree-of-freedom microcanonical ensemble. Second, the effect of the coupling is concentrated at the low end of the window. At $E=0.5$, switching on $b=0.3$ lowers the direct non-reactive fraction from $0.414$ to $0.382$, a decrease of $0.032\pm0.006$; the roaming non-reactive fraction rises by $0.030\pm0.004$ and the roaming reactive by $0.011\pm0.003$, while the direct-reactive fraction falls by $0.008\pm0.006$, within its sampling error. At $E=0.7$ the total roaming gain is $0.024\pm0.008$; from $E=1.0\ \mathrm{kcal\,mol^{-1}}$ the two couplings agree within the sampling uncertainty. The reason is kinematic: faster trajectories cross the roaming region in fewer periods of the roaming motion than the transverse instability needs to act. Third, the direct-reactive fraction is unchanged by the coupling within error at every energy ($|\Delta|=0.008\pm0.006$ at $E=0.5$), consistent with, though not determined by, the $b$-insensitivity of the FR1 action (Section~\ref{sec:results}). The zero-crossing non-reactive class is absent at $E=0.5$ and grows slowly with energy, reaching $0.024$ ($b=0$) and $0.029$ ($b=0.3$) at $E=2.0$; folding it into the roaming non-reactive count, as a pure parity rule would, would overstate the high-energy roaming fraction by up to a third.

\begin{figure}[tbp]
\centering
\includegraphics[width=\textwidth]{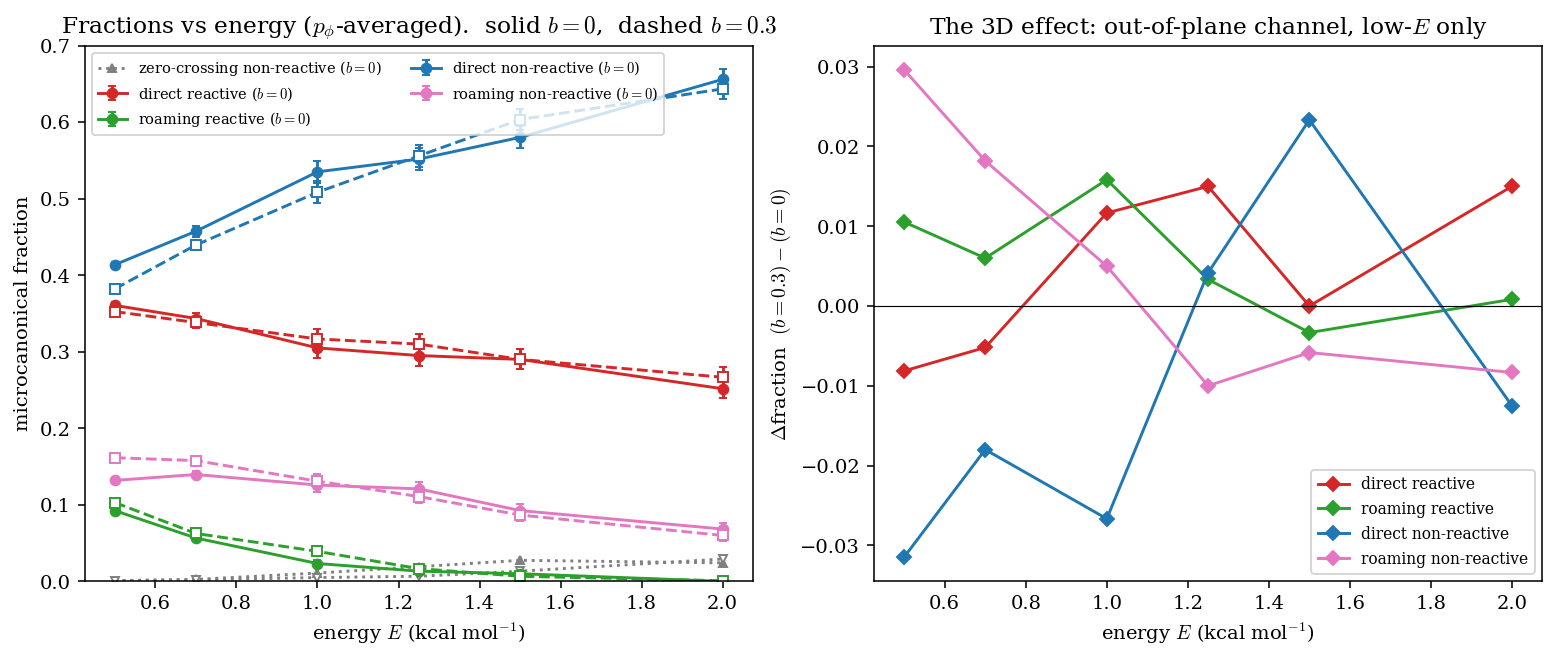}
\caption{Left: microcanonical class fractions versus energy for $b=0$ (filled symbols, solid lines) and $b=0.3$ (open symbols, dashed lines); error bars are binomial; the gray dotted curve is the zero-crossing non-reactive class. Sample sizes: $15000$ per coupling value at $E=0.5$, $5000$ at $E=0.7$, $1200$ otherwise. Right: the difference between the $b=0.3$ and the $b=0$ fraction of each class. The coupling transfers $0.040\pm0.005$ of the ensemble into the two roaming classes at $E=0.5$, $0.024\pm0.008$ at $E=0.7$, and nothing beyond the sampling uncertainty from $E=1.0\ \mathrm{kcal\,mol^{-1}}$. Same color code as Figure~\ref{fig:planes}.}
\label{fig:fractions}
\end{figure}

\subsection{Gap-time statistics}\label{sec:gaptimes}
The gap time of a non-reactive trajectory is the time between its entry through the entrance dividing surface and its exit back through it \cite{Mauguiere2014}. We measure it at the surface defined in Section~\ref{sec:protocol} --- the same placement as the planar study. The reactive analog, the time from entry to capture, is a different quantity and is reported in Appendix~\ref{app:classify}. Maugui\`ere, Collins, Ezra, Farantos, and Wiggins showed for the planar model that the gap-time distribution of the roaming region is not exponential: escape is not a memoryless, single-rate process \cite{Mauguiere2014}. The same conclusion, using the diagnostics applied below, is reached in \cite{Wiggins2026}. We take this as established, and ask how the coupling changes the distribution.

Figure~\ref{fig:gapbands} first shows the gap time along a one-parameter scan of the entrance momentum $p_\theta$ at fixed $\theta$ on the in-plane leaf. The direct bands sit at short, regular times, and the roaming trajectories cluster at the band edges with times that increase sharply toward the band boundaries --- the band structure of the planar roaming region \cite{Mauguiere2014}, recovered inside the three-degree-of-freedom model.

One geometric fact must be stated before distributions are compared. No gap can be shorter than the direct flight: a trajectory entering at $r_{\mathrm{OTS}}=13.386\ \AAA$ must travel inward to the interaction region and back out, and the purely radial flight to the shelf bottleneck and back takes $18.1$ time units. The shortest gap in the ensemble is $19.1$. The distribution therefore begins abruptly near this value, and that onset is a property of the geometry, not of the escape dynamics. For reference, the figures show an exponential distribution with the same mean, shifted to begin at the observed minimum; an unshifted exponential would put weight at gap times shorter than the flight itself. Because the coupling term is negligible beyond $8\ \AAA$ ($V_0<10^{-19}$ there), the inward and outward flights are identical for $b=0$ and $b=0.3$, and any difference between the two distributions arises in the interaction region.

Figure~\ref{fig:gapdist} shows the density and the survival function of the non-reactive gap time at $E=0.5$ and $E=1.0\ \mathrm{kcal\,mol^{-1}}$, for both couplings. Neither distribution is exponential. A convenient summary is the coefficient of variation --- the standard deviation divided by the mean, equal to $1$ for an exponential distribution and used for this purpose in \cite{Wiggins2026} --- computed on the gaps in excess of the minimum: at $E=0.5$ it is $1.56$ at $b=0$ and $1.22$ at $b=0.3$, both above $1$. The effect of the coupling is two-sided. Through the bulk of the distribution the gaps lengthen: the median rises from $30.2$ to $31.4$, the mean from $40.8$ to $41.3$, the 90th percentile from $69$ to $73$, and the fraction of trajectories with gap beyond $80$ rises from $7.3\%$ to $8.0\%$. These are the trajectories that wander out of the reaction plane before escaping --- the same population that swells the roaming non-reactive class. Conditioned on the roaming non-reactive class itself, the median gap is nearly unchanged ($51.4$ to $53.6$) and the mean falls ($69.5$ to $63.5$), so the shift of the pooled distribution reflects the larger roaming population rather than longer individual roaming gaps. In the extreme tail the ordering reverses: the 99th percentile falls from $193$ to $153$, with non-overlapping bootstrap $95\%$ confidence intervals $[181,209]$ and $[146,162]$. The rare, very long trappings of the symmetric limit are supported by conserved-$p_\phi$ structure that the coupling destroys. The coupling therefore shifts the pooled distribution toward longer typical gaps by enlarging the roaming population, leaves the typical roaming gap itself nearly unchanged, and cuts off the longest trappings; the region remains non-statistical at both couplings. At $E=1.0$ the two distributions are close (means $28.7$ and $30.6$), consistent with the fractions. Measuring the gap between crossings of an inner surface at $r=9.5\ \AAA$, which more than halves the deterministic flight, leads to the same conclusions (Appendix~\ref{app:classify}).

\begin{figure}[tbp]
\centering
\includegraphics[width=\textwidth]{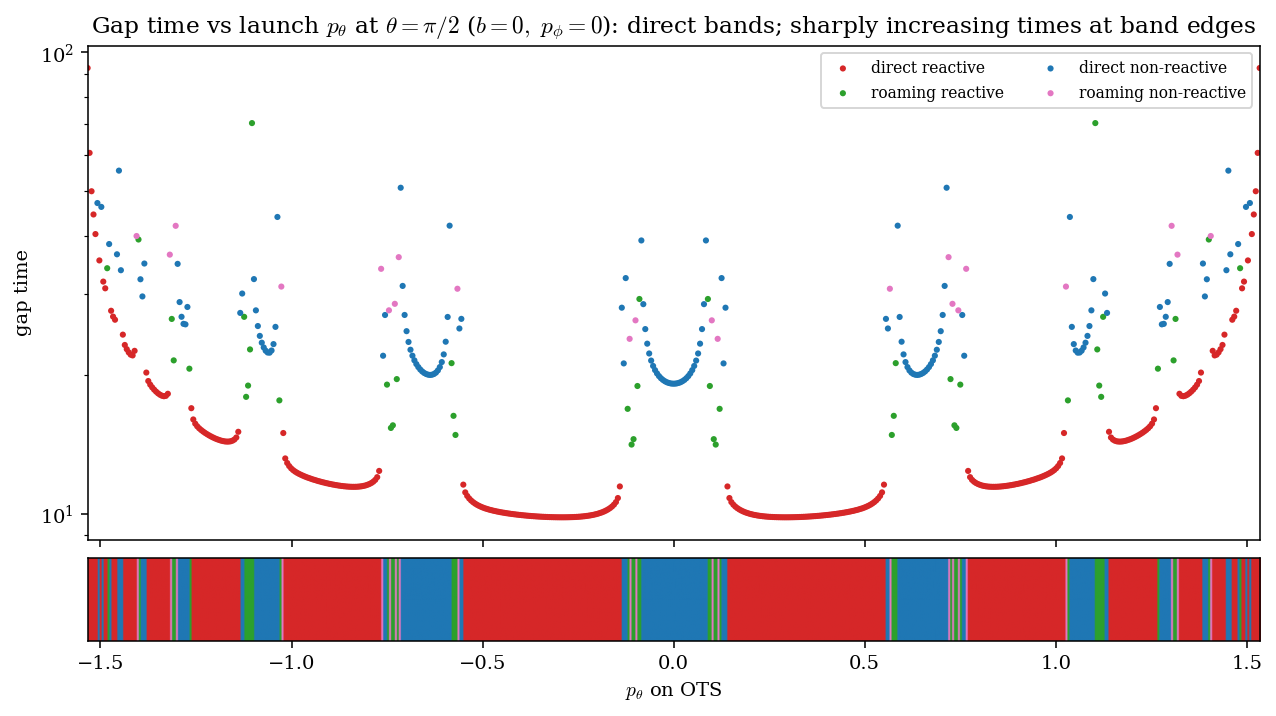}
\caption{Gap time versus the launch momentum $p_\theta$ at $\theta=\pi/2$ on the in-plane leaf ($b=0,\ p_\phi=0$), on a logarithmic scale: direct bands at short, regular times; roaming trajectories at the band edges with sharply increasing times. The strip beneath records the class. Same color code as Figure~\ref{fig:planes}.}
\label{fig:gapbands}
\end{figure}

\begin{figure}[tbp]
\centering
\includegraphics[width=\textwidth]{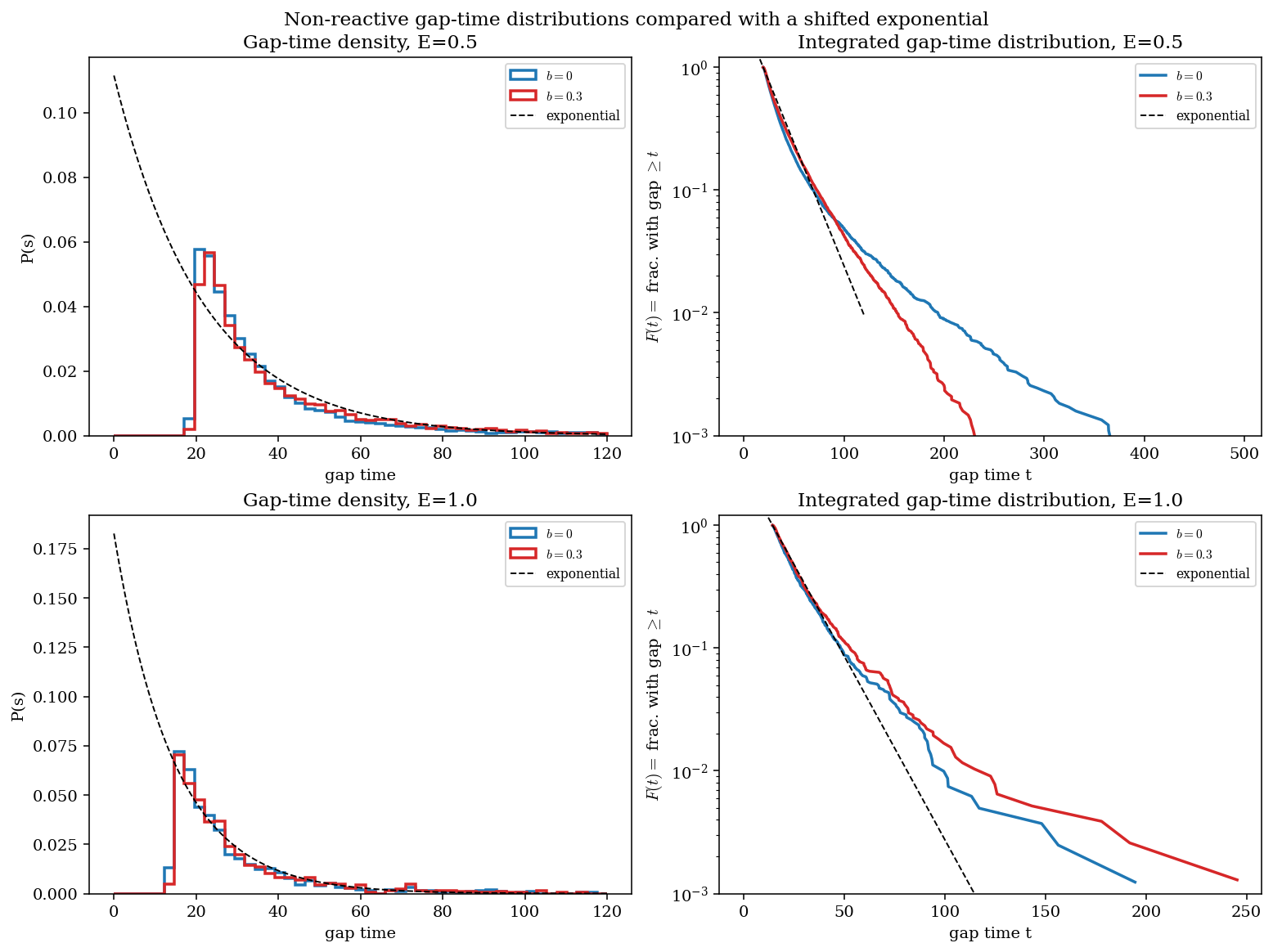}
\caption{Non-reactive gap times measured at the entrance dividing surface: density (left) and survival function (right, semi-logarithmic) at $E=0.5$ (top) and $E=1.0\ \mathrm{kcal\,mol^{-1}}$ (bottom), $b=0$ versus $b=0.3$. The dashed curve is an exponential distribution with the same mean, shifted to begin at the observed minimum gap (dotted vertical line), which is set by the flight time from the entrance surface to the interaction region and back. At $E=0.5$ the coupling shifts the body of the distribution toward longer gaps while shortening the extreme tail.}
\label{fig:gapdist}
\end{figure}

\subsection{The mechanism: transport of $p_\phi$ between the counter-precessing components}\label{sec:mechanism}
The mechanism behind the redistribution can be measured directly. Figure~\ref{fig:pphidiff} follows an ensemble of trajectories launched off the plane, at $p_\phi=0.8$, over a scan of the launch momentum $p_\theta$; for each trajectory we record the excursion of the azimuthal momentum along it, $\max p_\phi-\min p_\phi$. At $b=0$ the excursion is zero for every trajectory: $p_\phi$ is conserved, and each trajectory remains on its initial leaf. At $b=0.3$ the median excursion over the ensemble is $0.71$, the largest is $9.0$, and $22\%$ of the trajectories cross $p_\phi=0$ --- they reverse their sense of precession about the symmetry axis, passing between the two components of the manifold family that the symmetric limit keeps disjoint (Section~\ref{sec:families}). The trajectories thus show directly the mechanism established for the manifolds in Section~\ref{sec:persist}: for $b>0$ the conservation of $p_\phi$ fails, the conserved-momentum tori break up, and motion passes between the counter-precessing components.

\begin{figure}[tbp]
\centering
\includegraphics[width=0.8\textwidth]{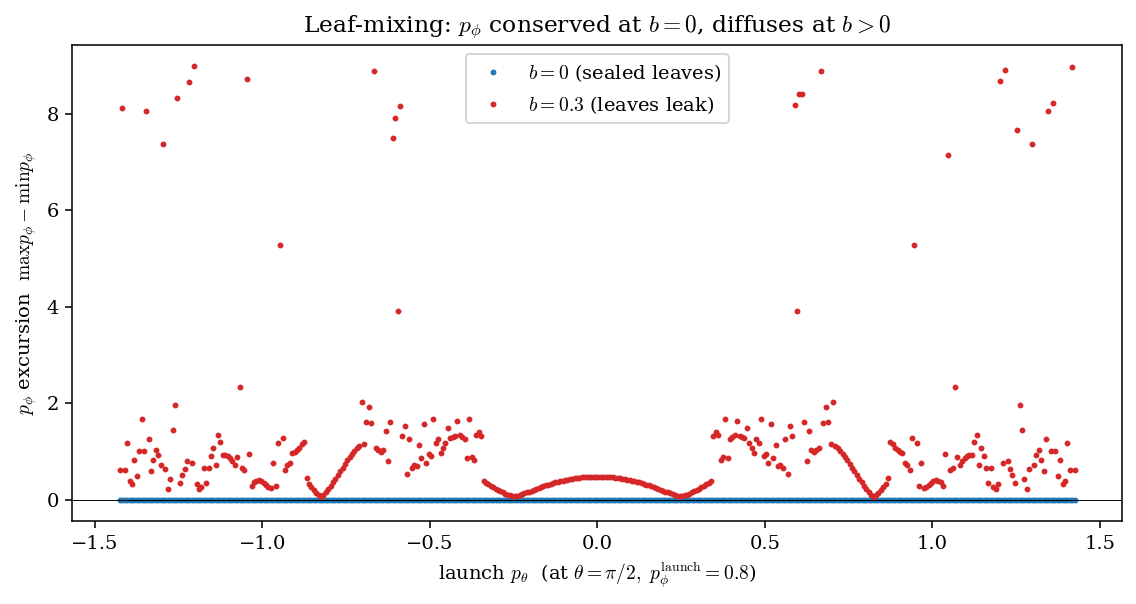}
\caption{The excursion of the azimuthal momentum, $\max p_\phi-\min p_\phi$, along each trajectory of an ensemble launched off the plane ($\theta=\pi/2$, $p_\phi=0.8$), versus the launch momentum $p_\theta$. At $b=0$ the excursion is zero for every trajectory. At $b=0.3$ the median excursion is $0.71$ and $22\%$ of the trajectories cross $p_\phi=0$, reversing their sense of precession about the symmetry axis.}
\label{fig:pphidiff}
\end{figure}

\subsection{Placement of the classification surfaces and robustness}\label{sec:robust}
Section~\ref{sec:nhim} constructs the three dividing surfaces as four-dimensional objects. The classification does not detect crossings of these surfaces directly; it uses fixed radii --- $13.386\ \AAA$ for entry and exit, $3.40\ \AAA$ for the classifier count, $2.0\ \AAA$ for capture --- and this subsection shows, in three steps, that the fixed radii represent the surfaces correctly and that the principal redistribution is robust over the tested range of surface placements.

The first step is the reduced family, computed at $b=0$ (Figure~\ref{fig:family}). The radial extent of the FR1 orbit varies little across the family: the inner turning radius moves from $3.18$ to $3.35\ \AAA$ and the outer from $3.65$ to $3.52\ \AAA$ over the admissible interval of $p_\phi$, so the radius $3.40\ \AAA$ lies inside the orbit's radial range on every leaf, and a single threshold serves the whole family. This is also why the threshold remains adequate at $b=0.3$, where $p_\phi$ changes along trajectories (Section~\ref{sec:mechanism}): for instantaneous $|p_\phi|$ within the admissible $b=0$ family the bottleneck remains at essentially the same radius, and trajectories leaving that range account for most of the observed classifier disagreements (Section~\ref{sec:robust}). The entrance orbit moves inward from $13.39\ \AAA$ at $p_\phi=0^+$ to $9.71\ \AAA$ at $p_\phi=2.10$, approaching the family endpoint (Section~\ref{sec:families}), so the launch surface at the planar value $13.386\ \AAA$ lies at or outside the entrance orbit of every leaf. One caveat is stated rather than assumed away. For small $p_\phi$, the centrifugal barrier of the effective potential $U$ in Eq.~\eqref{eq:Hred} lies beyond the launch surface --- at $25.6\ \AAA$ for $p_\phi=0.8$, moving inward as $p_\phi$ grows and crossing $13.386\ \AAA$ near $p_\phi\approx1.5$ --- so an outward crossing of the launch surface is not, for those trajectories, a crossing of their outermost barrier. The barrier is, however, very low: its top exceeds the potential at the launch surface by at most $0.0015\ \mathrm{kcal\,mol^{-1}}$ over the admissible interval at $E=0.5$. A trajectory that has crossed the launch surface outward can be turned back only if its outward radial kinetic energy there is below this value, and such trajectories are a negligible part of the ensemble.

The second step is direct measurement. Each trajectory is integrated once, the crossing times of all candidate surfaces are recorded, and the ensemble is reclassified under varied thresholds. At $E=0.5$, moving the classifier radius across the shelf ($3.25$--$3.55\ \AAA$) or the capture radius across the well mouth ($1.8$--$2.2\ \AAA$) shifts every class fraction by at most $0.011$ --- about a quarter of the $0.040$ effect of the coupling. Repeating the full computation with the launch surface at $14\ \AAA$ ($5000$ samples per coupling value) reproduces every fraction and the $E=0.5$ transfer within one standard error. The principal redistribution is robust over the tested range of surface placements.

The third step tests the radial classifier count against the dividing surface it stands in for. At $b=0$ the leaf of the FR1 dividing surface at angular momentum $p_\phi$ is the two-sphere over the reduced orbit (Section~\ref{sec:dstopo}); since the reduced orbit has its minimum radius at the equator and its maximum at the polar turning, its configuration arc is a graph $r=R(\theta;p_\phi)$, and a crossing of the surface is a sign change of $r-R(\theta;|p_\phi|)$ with $\theta$ inside the arc, the polar caps closing in momentum at the turning points. A passage of the radius $3.40\ \AAA$ near a pole, outside the arc, is not a crossing of the surface. Each trajectory of a subset of the $E=0.5$ ensemble ($2500$ per coupling value) was classified both ways from one integration: by the radial count and by crossings of the leafwise surface, evaluated at the instantaneous $|p_\phi|$; for $b=0.3$ the unperturbed $b=0$ surface serves as a frozen-$p_\phi$ reference classifier (no finite-$b$ surface is computed). The two classifications agree for $97.4\%$ of trajectories at $b=0$ and $94.5\%$ at $b=0.3$; the reactive classes agree for $99.8\%$. The disagreements concentrate in the trajectories whose $p_\phi$ exceeds the family endpoint $p_*=1.799$ on the shelf, where the $b=0$ surface is not defined ($65$ of $66$ disagreements at $b=0$, $123$ of $138$ at $b=0.3$); the remainder at $b=0.3$ are $13$ direct/roaming parity swaps that nearly cancel. Table~\ref{tab:dsval} gives the class fractions under both classifications. The reactive fractions are identical to three decimals at both couplings. The one systematic difference is that the surface count assigns a fraction of the ensemble ($0.026$ at $b=0$, $0.005$ at $b=0.3$) to the zero-crossing non-reactive class --- trajectories whose only crossings of the radius $3.40\ \AAA$ occur near the poles, outside every leaf surface --- a class the radial count leaves essentially empty at this energy (Section~\ref{sec:fractions}). The paired comparison is the test of the transfer: on this subset the radial estimate is noisier than the full-ensemble value ($+0.020$ against $+0.040\pm0.005$, with a subset standard error of $0.012$), and the same trajectories give a roaming gain of $+0.023$ under the surface count --- the two classifications agree on the transfer to within $0.003$ at the displayed precision. On this subset the frozen-$p_\phi$ reference classifier reproduces the roaming transfer obtained from the radial count within the sampling uncertainty, supporting the radial count as a robust realization of the transition-state geometry.

\begin{table}[tbp]
\centering
\caption{Class fractions from the same $2500$ trajectories per coupling value at $E=0.5\ \mathrm{kcal\,mol^{-1}}$, classified against the fixed radii (radial) and against the leafwise dividing surface. At $b=0$ the surface column uses the actual leafwise dividing surface; at $b=0.3$ it uses the unperturbed $b=0$ surface as a frozen-$p_\phi$ reference. The roaming gain is the change of $\mathrm{RR}+\mathrm{RN}$ from $b=0$ to $b=0.3$.}
\label{tab:dsval}
\begin{tabular}{lcccc}
\hline
 & radial, $b=0$ & surface, $b=0$ & radial, $b=0.3$ & frozen ref., $b=0.3$ \\
\hline
direct reactive & $0.351$ & $0.351$ & $0.341$ & $0.340$ \\
roaming reactive & $0.096$ & $0.096$ & $0.102$ & $0.102$ \\
direct non-reactive & $0.401$ & $0.383$ & $0.391$ & $0.392$ \\
roaming non-reactive & $0.151$ & $0.143$ & $0.165$ & $0.160$ \\
zero-crossing non-reactive & $0.000$ & $0.026$ & $0.000$ & $0.005$ \\
\hline
roaming gain & \multicolumn{2}{c}{$+0.020$ (radial)} & \multicolumn{2}{c}{$+0.023$ (surface)} \\
\hline
\end{tabular}
\end{table}

\begin{figure}[tbp]
\centering
\includegraphics[width=\textwidth]{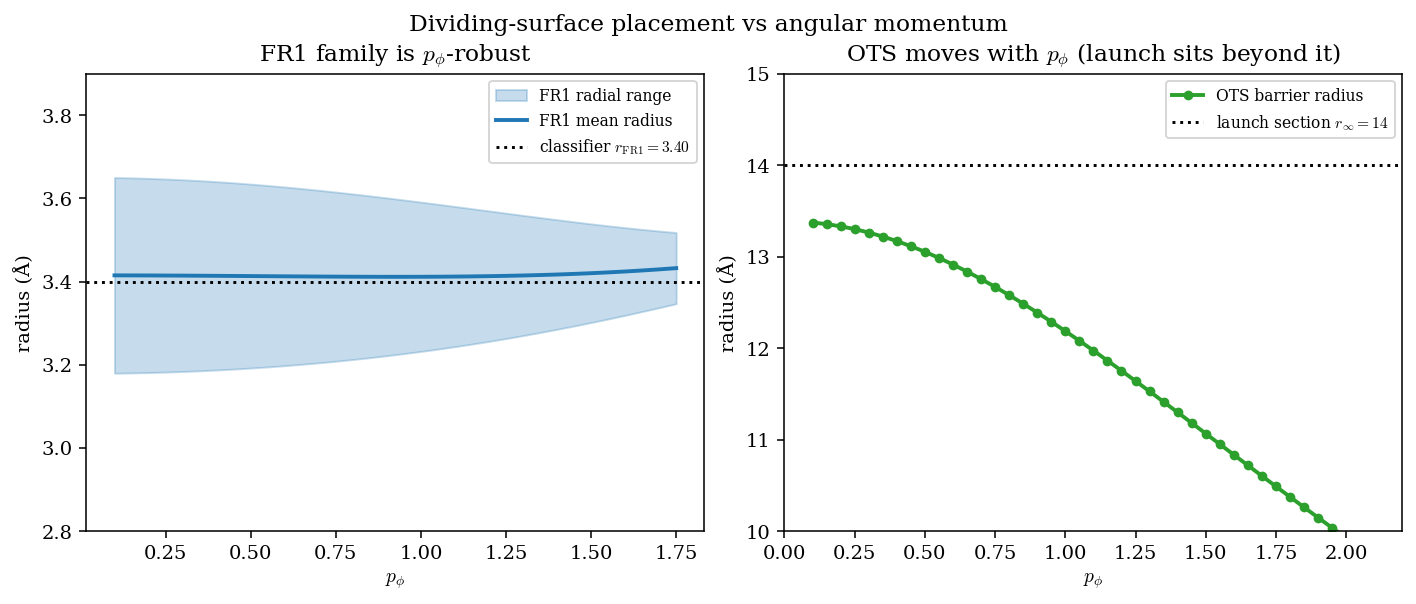}
\caption{Placement of the classification radii against the reduced family, computed at $b=0$. Left: the radial range and mean radius of the FR1 orbit versus $p_\phi$; the fixed classifier radius $3.40\ \AAA$ (dotted) lies inside the orbit's radial range on every leaf. Right: the radius of the entrance orbit versus $p_\phi$, from $13.39\ \AAA$ at the planar limit to $9.71\ \AAA$ near the family endpoint; the launch surface at the planar value (dotted) lies at or outside the entrance orbit of every leaf.}
\label{fig:family}
\end{figure}

\section{Discussion}\label{sec:discussion}
The analysis yields a coherent picture of how azimuthal corrugation reshapes roaming. The in-plane reactive bottleneck --- the FR1 action --- is essentially unaffected by the coupling, so the symmetry breaking does not change the rate at which trajectories pass through the roaming region in the reaction plane. What it changes is the transverse organization: breaking the cylindrical symmetry makes the out-of-plane direction hyperbolic, so that resolving the three-fold methyl corrugation opens an out-of-plane escape route to the roaming dynamics rather than confining motion to the plane. The transition is immediate rather than threshold-controlled: the out-of-plane direction is unstable for arbitrarily weak coupling and remains unstable across most of the range studied; after the narrow elliptic restabilization window $0.58\lesssim b\lesssim0.63$, the period-doubling at $b_c\approx0.63$ marks the onset of negative hyperbolicity. The objects that organize transport in the full space are not the one-dimensional orbits but three three-dimensional NHIMs, one for each transition state. At $b=0$, Section~\ref{sec:nhim} explicitly constructs the three NHIMs as two-component families of invariant tori and proves that every compact interior piece persists for sufficiently small coupling. Each anchors a four-dimensional dividing surface. At fixed $p_\phi$ the leaf of that surface is a copy of $S^2\times S^1$, and the corresponding torus of the NHIM cuts the leaf into two solid tori; this holds for all three transition states because all three reduced orbits librate. In the analysis of Section~\ref{sec:persist}, the mechanism of the escape route is the loss of conservation of $p_\phi$ once $b>0$: the separatrix structure at the center of each family is destroyed and the conserved-momentum tori break up.

Section~\ref{sec:transport} carries out the transport computation based on these surfaces. The orbit and manifold analysis of Sections~\ref{sec:results}--\ref{sec:nhim} is carried out at the single energy $E=0.5\ \mathrm{kcal\,mol^{-1}}$; the classification of Section~\ref{sec:transport} is in addition repeated across the roaming window $E=0.5$--$2.0\ \mathrm{kcal\,mol^{-1}}$, for one purpose: to locate where the effect of the coupling survives. The findings are as follows. The $p_\phi$-averaged microcanonical ensemble is far less reactive than the planar ($p_\phi=0$) slice --- at $E=0.5$ the direct-reactive fraction of the incoming ensemble is $0.36$, against $0.61$ on the planar leaf --- so the two-degree-of-freedom problem is not representative of the three-degree-of-freedom one. The coupling leaves the direct-reactive fraction unchanged within sampling error but redistributes the non-reactive population: at $E=0.5$, switching on $b=0.3$ lowers the direct non-reactive fraction from $0.414$ to $0.382$, a decrease of $0.032\pm0.006$; the roaming non-reactive fraction rises by $0.030$ and the roaming reactive by $0.011$, while the direct-reactive fraction falls by $0.008$, within sampling error; the pooled non-reactive gap distribution shifts toward longer typical times as the roaming population grows, while the longest trappings of the symmetric limit are cut off. Across the window the effect decays --- the transfer is $0.024$ at $E=0.7$ and zero within sampling error from $E=1\ \mathrm{kcal\,mol^{-1}}$ --- because faster trajectories traverse the roaming region in fewer periods of the roaming motion than the transverse instability needs. The trajectories show the same mechanism that Section~\ref{sec:persist} establishes for the manifolds: $p_\phi$ is no longer constant along trajectories, and about a fifth of the off-plane launches reverse their sense of precession about the symmetry axis, passing between the two components that the symmetric limit keeps disjoint (Section~\ref{sec:mechanism}).

The relation to the ozone recombination study of Maugui\`ere, Collins, Kramer, Carpenter, Ezra, Farantos, and Wiggins \cite{Mauguiere2016} --- a three-degree-of-freedom roaming study in the same phase-space framework --- is complementary. There the third degree of freedom, the diatomic vibration, decouples adiabatically, and the three-degree-of-freedom computation validates the two-degree-of-freedom reduction; here the third degree of freedom is activated by the symmetry-breaking coupling, and the reduction fails in a controlled, quantifiable way. Both are families of two-degree-of-freedom subsystems parametrized by a momentum --- conserved there in an adiabatic approximation, conserved here exactly at $b=0$ and destroyed for $b>0$. Two limitations delimit the present study. The analysis is carried out at zero total angular momentum; the rotating ($J\neq0$) case introduces Coriolis coupling and centrifugal terms that enlarge the accessible phase space and may support trapping mechanisms absent here. And we have not examined behavior near the upper edge of the roaming energy window.

\section{Conclusion}\label{sec:conclusion}
We have constructed and analyzed a concrete, tractable three-degree-of-freedom Chesnavich model designed for the phase-space analysis of roaming. The model is derived from the Ezra--Wiggins rigid-body formulation, breaks the cylindrical symmetry of Chesnavich's model with a coupling that respects the physical $C_3$ symmetry of the methyl fragment, uses the physically mandated planar-top inertia ratio $I_z=2I_x$, and recovers the 2-DoF model at $b=0$ as its planar reduction.

Its central dynamical feature is the immediate activation of the out-of-plane degree of freedom when the cylindrical symmetry is broken. With the physical inertia ratio and the $C_3$ coupling, the roaming-shelf orbit FR1 is transversely hyperbolic for essentially all $b>0$: below the narrow elliptic window $0.58\lesssim b\lesssim0.63$ the transverse multiplier pair is real and positive; beyond the period-doubling at $b_c\approx0.63$ it is real and negative, so that neighboring trajectories alternate sides of the reaction plane on successive periods. FR1 is positive-hyperbolic below the elliptic window, elliptic within it, and negative-hyperbolic beyond the period-doubling at $b_c$.

A dividing surface in the five-dimensional energy surface must be four-dimensional, and it must be anchored on a three-dimensional invariant manifold --- two dimensions more than a periodic orbit provides. We therefore identified and explicitly constructed the anchoring objects: three three-dimensional normally hyperbolic invariant manifolds, one for each transition state, exhibited in the $b=0$ limit as two-component families of invariant tori over the conserved angular momentum $p_\phi$. Three results about these manifolds are established in Sections~\ref{sec:nhim} and~\ref{sec:proofs}, as six propositions resting on a single numerically established hypothesis: the topology of the dividing surface each manifold anchors; the status of each generating orbit as the separatrix member of its family; and the persistence of every compact interior piece of each manifold for sufficiently small $b>0$.

Building on these surfaces, we carried the classification of Maugui\`ere, Collins, Ezra, Farantos, and Wiggins into three degrees of freedom (Section~\ref{sec:transport}). At $b=0$ the dynamics decomposes into a one-parameter family of two-degree-of-freedom subsystems, and these subsystems differ substantially from one another: the roaming fraction of a leaf grows with its angular momentum, doubling between $p_\phi=0$ and $p_\phi=0.8$. Breaking the symmetry opens an out-of-plane escape route. At the base energy $E=0.5\ \mathrm{kcal\,mol^{-1}}$ it moves $0.032$ of the incoming ensemble out of direct non-reactive escape, the two roaming classes gaining $0.040$, and shifts the pooled gap-time distribution toward longer typical times through the larger roaming population, while leaving the direct-reactive fraction unchanged; the effect dies away across the roaming energy window. Along trajectories the route appears as the loss of constancy of $p_\phi$: about a fifth of the off-plane launches reverse their sense of precession about the symmetry axis, passing between the two family components that the symmetric limit keeps disjoint. The model provides a concrete, tractable setting in which the three-degree-of-freedom phase-space theory of roaming can be developed and tested.

\appendix
\section{Symbolic derivation and verification}\label{app:symbolic}
The reduced Hamiltonian \eqref{eq:H} and the properties (i)--(iii) of Section~\ref{sec:model} were verified with a computer-algebra system. The rotational kinetic energy was checked against the Ezra--Wiggins reduction \cite{EzraWiggins2019}. The $b=0$ reduction to the 2-DoF Chesnavich Hamiltonian was verified to vanish identically in the difference of the two expressions. The discrete symmetries were verified by direct substitution. The Cartesian form of the coupling, $V_0[(X^2+Y^2)/r^2 + b(X^3-3XY^2)/r^3]$, was confirmed to be smooth away from the origin and to reproduce \eqref{eq:Vcoup} under $(X,Y,Z)=r(\sin\theta\cos\phi,\sin\theta\sin\phi,\cos\theta)$.
\section{The body-frame coordinate singularity}\label{app:singularity}
Body-frame spherical coordinates $(r,\theta,\phi)$ are singular on the symmetry axis $\theta=0,\pi$, where $\phi$ is undefined. Two distinct difficulties arise there, and both must be handled because every generating orbit crosses the axis twice per period: the planar FR1 and OTS orbits rotate, passing through both poles each period, and the planar TTS orbit librates, crossing a single pole twice per period (the axis is met to within $\sim2\times10^{-4}\ \mathrm{rad}$ in our computations).

First, the kinetic coefficient of the azimuthal variational equation contains a factor $\csc^2\theta$, which diverges on the axis. A transverse stability computation carried out naively in $(\phi,p_\phi)$ therefore returns spurious, integrator-tolerance-dependent multipliers as large as $10^9$. The resolution is to compute the transverse dynamics in Cartesian body-frame coordinates $(X,Y,Z)$, in which the out-of-plane pair $(Y,p_Y)$ is smooth across the axis; the transverse block of the monodromy matrix is then finite and well behaved (the values reported in Fig.~\ref{fig:trace}).

Second, and more fundamentally, the potential coupling must itself be smooth across the axis for the dynamics to be well defined there. A one-fold coupling of the form $\sin^2\theta\cos\phi$ has Cartesian form proportional to $X\sqrt{X^2+Y^2}/r^2$, which is continuous but only once differentiable on the axis; its second derivatives, which enter the variational equation, are discontinuous. This is why the $C_3$ coupling $\sin^3\theta\cos3\phi=(X^3-3XY^2)/r^3$, the angular factor of a solid harmonic and smooth away from the origin, is required: it is both physically appropriate (respecting the methyl three-fold symmetry) and analytically smooth across the axis the orbits traverse. The axis crossing also underlies the conical center of each two-component NHIM family (Section~\ref{sec:separatrix}): the closest approach of the FR1 family to the axis scales linearly with the angular momentum, $\rho_{\min}\simeq1.69\,|p_\phi|$, so the family meets the axis in a cone rather than tangentially; the OTS and TTS families have their own linear constants (Section~\ref{sec:separatrix}).

\section{Periodic-orbit location, continuation, and stability}\label{app:po}
Trajectories were integrated with the explicit adaptive Runge--Kutta method of order eight, DOP853 \cite{HairerNorsettWanner1993}, with relative and absolute tolerances of $10^{-12}$ and $10^{-13}$; energy was conserved to one part in $10^{12}$ or better along all reported orbits. The FR1 and OTS orbits were located in the reaction plane by shooting from the symmetric equatorial launch ($\theta=\pi/2$, $p_r=0$) and imposing the symmetric-return condition $p_r=0$ at the first axis crossing. The launch radius is determined by this condition: $r_0$ is a root of the scalar function $r_0\mapsto p_r(t_{\mathrm{axis}};r_0)$, located by bisection. At $p_\phi=0$ the FR1 and OTS orbits rotate --- $\theta$ advances monotonically and $p_\theta$ has no zeros --- so for these two orbits a $\theta$-turning cannot serve as the return event. The shooting returns FR1 at $r_0\approx3.18$ and OTS at $r_0\approx13.4$. The TTS orbit crosses the axis rather than the equator and was located by shooting from the axis ($\theta=0$, $p_r=0$), with the return condition $p_r=0$ imposed at the first zero of $p_\theta$: it librates through the axis with $\theta_{\max}\approx41^\circ$ over $r\in[2.38,2.64]$, arcing across the mouth of the well just outside the ridge crest $r_c=2.2$, with period $T=2.228$ and multiplier $\lambda_{\mathrm{TTS}}=220.5$. The same axis-shooting condition admits three further periodic orbits interior to the well, with axis crossings at $r_0\approx0.98$, $1.19$, and $1.52$ and $\theta_{\max}\approx49^\circ$, $67^\circ$, and $50^\circ$; continuation steps that are too large can cause the continuation to jump onto these, and the orbit at $r_0\approx1.19$ loses hyperbolicity entirely for $b\gtrsim0.05$, so a jump is detectable by a collapse of the multiplier. Each orbit was continued in $p_\phi$, in $b$, and in $E$ by using the solution at one parameter value as the initial guess at the next~\cite{AllgowerGeorg1990}. The monodromy matrix was obtained by integrating the $6\times6$ variational system over one period with the analytic Jacobian~\cite{ParkerChua1989,MeyerHallOffin2009}; the transverse block was extracted in the $(Y,p_Y)$ coordinates, and the elliptic/hyperbolic classification follows from whether its trace lies within or outside $[-2,2]$. The bifurcation threshold was located by bracketing the transverse trace to the value $-2$, giving $b_c=0.632$ both from the variational monodromy matrix and from a finite-difference monodromy matrix of the time-$T$ flow, stable under tightening the integrator tolerance from $10^{-11}$ to $10^{-13}$. The equatorial shooting condition admits a second periodic-orbit solution besides FR1, with action $\approx12.6$; FR1 is the solution with action $13.755$, identified by its reduction to the 2-DoF orbit at $b=0$, and continuation steps of $\Delta b\lesssim0.03$ keep the continuation on it.

\section{Trajectory classification: sampling, gap times, and robustness}\label{app:classify}
Trajectories for the classification of Section~\ref{sec:transport} were integrated with the same DOP853 scheme as the orbits (Appendix~\ref{app:po}), at relative and absolute tolerances $10^{-10}$ and $10^{-12}$, with the reduced mass $m=0.9445$ (from $m_{\mathrm{H}}=1.007825$ and $m_{\mathrm{C}}=12.0$); energy was conserved to better than one part in $10^{8}$, ample for counting surface crossings. The classification was performed for two values of the coupling, $b=0$ and $b=0.3$, and repeated across the roaming energy window for one purpose: to locate where the effect of the coupling dies away. Sample sizes: $15000$ per coupling value at $E=0.5$, $5000$ at $E=0.7$, and $1200$ at each of $E\in\{1.0,1.25,1.5,2.0\}\ \mathrm{kcal\,mol^{-1}}$; the classification planes of Fig.~\ref{fig:planes} used an $80\times110$ grid in $(\theta,p_\theta)$ at $\phi=0$. All other computations in the paper are at $E=0.5\ \mathrm{kcal\,mol^{-1}}$.

Initial conditions were sampled on the entrance surface at $r_{\mathrm{OTS}}=13.386\ \AAA$, uniformly in the canonical surface coordinates $(\theta,\phi,p_\theta,p_\phi)$ over the energetically allowed region with $\theta\in(0.05,\pi-0.05)$, and $p_r<0$ fixed by $H=E$. The omitted polar caps carry $\approx1.0\times10^{-3}$ of the surface measure, computed by quadrature of the allowed $(p_\theta,p_\phi)$ area over $\theta<0.05$ --- an order of magnitude below the smallest class shift reported in Section~\ref{sec:transport}. This uniform sample realizes the flux-weighted microcanonical measure --- the ensemble in which initial conditions on the surface are weighted by the rate at which trajectories cross it, the appropriate ensemble for an incoming stream of trajectories. The identification holds because the crossing-rate factor cancels the Jacobian of the energy constraint: on the section $\{r=r_0\}$ the magnitude of the directional flux is
\begin{equation*}
\mathcal{F}(E)=\int\delta(H-E)\,\delta(r-r_0)\,|\dot r|\,\Theta(-\dot r)\,d\mathbf{q}\,d\mathbf{p}
=\int_{\mathrm{allowed}}d\theta\,d\phi\,dp_\theta\,dp_\phi,
\end{equation*}
with $|\dot r|=|p_r|/m$ canceling the Jacobian of $\delta(H-E)$ on integrating out $p_r$.

In the computation the three surfaces are realized as radial thresholds: a crossing of the classifier is a crossing of $r_{\mathrm{FR1}}=3.40\ \AAA$; a trajectory is reactive when $r$ falls below $r_{\mathrm{react}}=2.0\ \AAA$ and non-reactive when it returns outward through the launch surface. Classes follow the parity rule of Section~\ref{sec:protocol}, with zero-crossing non-reactive trajectories reported as their own class. No trajectory remained unclassified at the integration limit $t_{\max}=500$.

The gap time of Section~\ref{sec:gaptimes} is measured at the entrance dividing surface, between entry and exit. Two times were in fact recorded for every trajectory: the gap at the entrance surface, and the time between the first inward and last outward crossing of an inner surface at $r=9.5\ \AAA$. The purely radial flights from the two surfaces to the shelf bottleneck and back take $18.1$ and $10.6$ time units respectively at $E=0.5$ (the shortest observed gaps are $19.1$ and $11.6$); the corrugation satisfies $V_0(r)<10^{-19}$ for $r>8\ \AAA$, so the inward and outward flights are identical for the two couplings, and any $b$-dependence of the statistics arises inside the interaction region. The distributions of Section~\ref{sec:gaptimes} are conditioned on the non-reactive outcome. The reactive capture times, measured from entry at the entrance surface to arrival below the capture radius, have means $19.0$ ($b=0$) and $20.4$ ($b=0.3$) at $E=0.5$. The inner-surface gap statistics lead to the same conclusions as the entrance-surface statistics reported in the text: at $E=0.5$ the non-reactive means are $25.5$ ($b=0$) and $26.2$ ($b=0.3$), the coefficients of variation beyond the minimum are $1.9$ and $1.4$, the 90th percentile rises from $41$ to $47$, and the extreme tail shortens.

The robustness scan integrates each trajectory once, recording the crossing times of all candidate surfaces, and reclassifies. At $E=0.5$, over the grid $r_{\mathrm{FR1}}\in\{3.25,3.40,3.55\}\ \AAA$ and $r_{\mathrm{react}}\in\{1.8,2.0,2.2\}\ \AAA$, every class fraction moves by at most $0.011$, about a quarter of the $0.040$ out-of-plane effect; repeating the full computation with the launch surface at $14\ \AAA$ ($5000$ samples per coupling value) reproduces every fraction and the $E=0.5$ transfer within one standard error. The centrifugal barrier that lies beyond the launch surface for small $p_\phi$ (Section~\ref{sec:robust}) exceeds the potential at the launch radius by at most $0.0015\ \mathrm{kcal\,mol^{-1}}$ over the admissible $p_\phi$, so only trajectories leaving the launch surface with outward radial kinetic energy below $0.0015\ \mathrm{kcal\,mol^{-1}}$ can be turned back --- a negligible subset of the ensemble.

\section{Numerical verification}\label{app:verification}
The results were re-derived by methods independent of those that produced them, and all checks pass.

For the planar orbits: FR1 was recovered as a fixed point of a Poincar\'e return map on the section $\{Z=0,\,p_Z>0\}$ (residual $\sim10^{-13}$). This computation imposes no symmetry: the launch is not restricted to $p_r=0$ on the equator, yet the converged fixed point satisfies that condition, independently confirming the symmetric construction of Appendix~\ref{app:po}. The transverse block of the monodromy matrix from the analytic Jacobian agrees with a finite-difference monodromy matrix of the time-$T$ flow to $\sim10^{-12}$ in the trace at every value of $b$ in the continuation grid of Fig.~\ref{fig:trace}; the full monodromy matrix is symplectic ($\det=1$) with reciprocal eigenvalue pairs, the hyperbolic pair being $\{24.57,\ 0.0407\}$. The abbreviated action computed as $\oint 2(E-V)\,dt$ agrees with $\oint p\,dq$ at every $b$ in the same grid. The TTS orbit was additionally verified by a computation that involves no trajectory integration at all: a discrete-variational solution of Hamilton's principle on a periodic lattice, with the multiplier obtained from the block product of the discrete-action Hessian. Along the TTS family the abbreviated action $W=\oint p\,dq$ and the period $T$ were checked against the classical identity $dW/dE=T$ (the derivative taken along the family at fixed $p_\phi$), and the integration-free and shooting routes agree to five significant figures.

For the two-component manifolds of Section~\ref{sec:nhim} the following were verified independently. (i) The azimuthal momentum is conserved to $10^{-11}$ on a generic trajectory at $b=0$. (ii) The reduced family is normally hyperbolic with the algebraic multiplicity of the eigenvalue $1$ equal to four and the tangent directions (flow, azimuthal rotation, the family direction $\partial/\partial p_\phi$) carrying no component in the hyperbolic eigenspace (spectral-projector residuals $\lesssim10^{-13}$); the normal multiplier of FR1 decreases monotonically from $24.6$ to unity at $p_*=1.799$, $\lambda$ real and greater than one at every interior $p_\phi$. (iii) The closest approach of the FR1 family to the axis scales \emph{linearly} with the angular momentum, $\rho_{\min}=1.686\,|p_\phi|$ to four figures over $0<p_\phi<0.4$ --- a cone, not a smooth tangency. The measured slope agrees with the prediction of Section~\ref{sec:separatrix}: with $\theta_-=k\,|p_\phi|$ the linear law for the polar turning angle and $r_t$ the radius there, the closest approach is $\rho_{\min}\approx r_t\sin\theta_-\approx r_t\,k\,|p_\phi|$, and $r_t\,k=1.6863$. (iv) Along the FR1 family the kinetic energy at the polar turning point tends to a finite nonzero limit, $E-V\to1.1736$ as $p_\phi\to0^+$, carried entirely by the azimuthal motion --- the signature of the momentum segments in the limit of the family (Section~\ref{sec:separatrix}). (v) The reduced dividing surface is a two-sphere for all three orbits (each reduced orbit librates), and the flux form \eqref{eq:flux} vanishes only at the two periodic-orbit momenta on each momentum ellipse. (vi) The invariance of the $b=0$ tori and their breakup for $b>0$ were checked directly on a single trajectory started on an interior torus and followed for twelve periods of the roaming motion. At $b=0$ its $p_\phi$ is constant to $10^{-11}$, as it must be on an invariant torus. At $b=0.05$ the value of $p_\phi$ along the same trajectory ranges over an interval of width $0.136$, and at $b=0.30$ of width $1.02$; energy is conserved to $\sim10^{-9}$ throughout, so the spread is dynamical, not numerical --- the tori are gone.

The quantities established for the reduced families in Section~\ref{sec:families} are load points of the paper: the endpoint behavior of each family and the nondegeneracy conditions enter the propositions of Section~\ref{sec:proofs} as hypotheses, and the family radii enter Section~\ref{sec:transport} as the placement data for the classification surfaces. Each was therefore re-derived independently of the continuation that produced it, with the machinery first validated on FR1 against the $\lambda(p_\phi)$ curve and the endpoint $p_*=1.799$. The termination mechanisms: the OTS family terminates at $p_*=2.1495$ by amplitude collapse onto the azimuthal relative equilibrium at $r=9.600$, and the TTS family at $p_*=1.3168$ on the saddle--center equilibrium that the centrifugal term creates in the mouth of the well ($r=2.446$, $\theta=29.4^\circ$) --- in both cases the endpoint multiplier is predicted from the linearized rates of the limiting equilibrium, $e^{\mu_r\cdot2\pi/\omega_\theta}$, giving $1.2685$ and $15.90$, and matched by the family; this confirms that these two families end on caps with normal hyperbolicity intact, as the per-family statements of Sections~\ref{sec:families} and~\ref{sec:dstopo} require. The TTS inner-edge limits $T(0^+)=T_{\mathrm{planar}}/2$ and $\lambda(0^+)=\sqrt{\lambda_{\mathrm{planar}}}$ confirm the double-cover relation between the reduced TTS family and the planar orbit. The linear law $\theta_-=k\,|p_\phi|$ holds per family with $k$ given by the turning-point formula of Section~\ref{sec:separatrix} (FR1 $0.4619$, OTS $0.651$, TTS $0.304$), confirming the cone constants. And the two nondegeneracy hypotheses of Proposition~\ref{prop:nhim} were verified directly: the azimuthal advance $\Delta\phi(p_\phi)$ is monotone (from $2\pi$ to $4.613$ across the FR1 family) and $dT/dE=-3.16\neq0$ at $p_\phi=0.5$.

For the trajectory classification (Section~\ref{sec:transport}): the integrating core reproduces the published planar FR1 values at $m=0.9445$ ($r_0=3.1791$, $W=13.7547$, $\lambda=24.5732$, transverse trace $+2$ to $2\times10^{-11}$); the planar classification leaf reproduces the two-degree-of-freedom fractions ($88.1\%$ direct, $11.9\%$ roaming); the $E=0.5$ transfer is established at $15000$ samples per coupling value, the direct non-reactive and roaming non-reactive shifts significant at six and seven standard errors, the roaming reactive shift at three, and the direct-reactive shift within $1.5$ standard errors; and the threshold and launch-radius scans of Appendix~\ref{app:classify} bound the placement sensitivity at a quarter of the effect. Finally, the planar TTS orbit, both family caps and endpoint multipliers, the FR1 endpoint, the cone constants, and the double-cover limits were reproduced by a second, fully disjoint implementation --- potential and derivatives regenerated from the published formulas, a fixed-step implicit-midpoint symplectic integrator in place of DOP853, multipliers from finite-difference monodromy matrices, cap predictions by root-finding and linearization without integration --- together with the classical identities $dW/dE=T$ and $dW/dp_\phi=-\Delta\phi$, satisfied to $10^{-6}$ or better.

\end{document}